\pdfoutput=1
\documentclass[final]{IEEEtran}

\usepackage{cite}
\usepackage[pdftex]{graphicx}
\usepackage{epstopdf}
\usepackage{amssymb,amsfonts,amsmath,stmaryrd,amsthm}
\usepackage[ruled,vlined]{algorithm2e}
\usepackage{array}
\usepackage[caption=false,font=normalsize,labelfont=sf,textfont=sf]{subfig}
\usepackage{dblfloatfix}
\usepackage{ulem}
\usepackage{url,color}
\usepackage[pdftex]{hyperref}
\newcommand{\new}[1]{#1}

\newcommand*{\hermconj}{^{\mathsf{H}}}
\renewcommand{\top}{{\mathsf{T}}}

\def\RR{{\mathbb R}}

\def\EE{{\mathbb E}}

\def\CC{\mathbb C}

\DeclareMathAlphabet{\mathpzc}{OT1}{pzc}{m}{it}

\def\bA{{\mathbf A}}
\def\bb{{\mathbf b}}

\def\bv{{\mathbf v}}

\def\b1{\boldsymbol{1}}

\def\b{{\underline{b}}}

\def\v{{\underline{v}}}

\newcommand{\diag}[1]{\mbox{diag}(#1)}
\newcommand{\ve}[1]{ {\mathbf{#1}} }
\newcommand{\ul}[1]{ {\underline{#1}} }

\def\defequal{\stackrel{\mbox{\footnotesize def}}{=}}

\newcommand{\bean}{\begin{eqnarray*}}
\newcommand{\eean}{\end{eqnarray*}}

\newcommand{\W}{\ve{W}}
\renewcommand{\H}{\ve{H}}
\def\bal#1{\begin{align}#1\end{align}}

\def\aa{\text{a}}
\def\bb{\text{b}}
\newcommand{\balpha}{\boldsymbol{\alpha}}
\newcommand{\bualpha}{\underline{\boldsymbol{\alpha}}}
\newcommand{\bphi}{\boldsymbol{\phi}}
\newcommand{\bPhi}{\boldsymbol{\Phi}}
\newcommand{\buPhi}{\underline{\boldsymbol{\Phi}}}
\newcommand{\bSigma}{\boldsymbol{\Sigma}}
\newcommand{\btheta}{\boldsymbol{\theta}}

\def\IS{\text{IS}}
\def\JL{\text{JL}}

\def\EE{\mathbb{E}}

\newtheorem{theorem}{Theorem}

\begin{document}

	\title{Estimation with Low-Rank Time-Frequency Synthesis Models}	
	
	\author{C\'edric~F\'evotte,~\IEEEmembership{Senior Member,~IEEE,}
	and~Matthieu~Kowalski% <-this % stops a space
	\thanks{C.~F\'evotte is with IRIT, Universit\'e de Toulouse, CNRS, France.}% <-this % stops a space
	\thanks{M.~Kowalski is with L2S, Universit\'e Paris-Sud, Gif-sur-Yvette, France.}% <-this % stops a space
	}

	\maketitle

	\begin{abstract}
	
		Many state-of-the-art signal decomposition techniques rely on a low-rank factorization of a time-frequency (t-f) transform. In particular, nonnegative matrix factorization (NMF) of the spectrogram has been considered in many audio applications. This is an {\em analysis} approach in the sense that the factorization is applied to the squared magnitude of the analysis coefficients returned by the t-f transform. In this paper we instead propose a {\em synthesis} approach, where low-rankness is imposed to the synthesis coefficients of the data signal over a given t-f dictionary (such as a Gabor frame). As such we offer a novel modeling paradigm that bridges t-f synthesis modeling and traditional analysis-based NMF approaches. The proposed generative model allows in turn to design more sophisticated multi-layer representations that can efficiently capture diverse forms of structure. Additionally, the generative modeling allows to exploit t-f low-rankness for compressive sensing. We present efficient iterative shrinkage algorithms to perform estimation in the proposed models and illustrate the capabilities of the new modeling paradigm over audio signal processing examples.
	\end{abstract}

	\IEEEpeerreviewmaketitle

	\section{Introduction}

	\IEEEPARstart{M}{atrix}  factorization methods currently enjoy a large popularity in machine learning and signal processing. In signal processing, the input data is usually a time-frequency (t-f) transform of some original time series $x(t)$. For example, in the audio setting, nonnegative matrix factorization (NMF) is commonly used to decompose magnitude or power spectrograms into elementary components \cite{Smaragdis2014}; the spectrogram $\ve{P}$ is approximately factorized into $\W \H$, where $\W$ is the dictionary matrix collecting spectral patterns in its columns and $\H$ is the activation matrix. The approximate $\W \H$ is generally of lower rank than $\ve{P}$, unless additional constraints are imposed on the factors. NMF is at the core of \new{classical} source separation systems such as \cite{vir07,ieee_asl10}. 
		
	The spectrogram $\ve{P}$ is usually obtained from the short-time Fourier transform $\ve{Y}$. The coefficients $y_{fn}$ of $\ve{Y}$ are the inner products of $x(t)$ with t-f atoms $\phi_{fn}(t)$, where $f$ and $n$ index frequencies and time frames, respectively, and a common choice is $\ve{P} = | \ve{Y} |^{2}$. The STFT coefficients are so-called {\em analysis} coefficients. As such, spectral decomposition by NMF can be viewed as a low-rank time-frequency {\em analysis} procedure. Leveraging on the potential of {\em synthesis} models as opposed to analytical ones (see, e.g., \cite{elad2007analysis,balazs2013adapted,sprechmann2013supervised,nam2013cosparse}), we propose to explore a dual view of the usual NMF approach and present a new paradigm that we name {\em low-rank time-frequency synthesis} (LRTFS). In this new paradigm, the signal is decomposed as
	\bal{ \label{eq:synth}
	x(t) = \sum_{fn} \alpha_{fn} \phi_{fn}(t) + e(t)
	}
	where the synthesis coefficients $\{\alpha_{fn}\}$ are endowed with a low-rank structure such that $|\alpha_{fn}|^{2} \approx [\W \H]_{fn}$. Formulation~\eqref{eq:synth} provides a generative representation of the raw data $x(t)$ and extends the modeling capacities of standard NMF-based signal decomposition towards more advanced multi-layer hybrid decompositions. Having a generative model of the raw data (instead of its transform) is also useful for some inverse problems such as compressive sampling, an application that will be considered in the paper.

	The low-rankness of the synthesis coefficients $\{ \alpha_{fn} \}$ is induced through a probabilistic model named {\em Gaussian Composite Model} (GCM) \cite{neco09}. The GCM underlies Itakura-Saito NMF, a baseline method that will be recalled in Section~\ref{sec:isnmf}. Section~\ref{sec:lrtfs-comp} presents our new paradigm LRTFS in the general case of complex-valued signals. It also describes an alternate minimization algorithm for maximum joint likelihood estimation of the parameters. Section~\ref{sec:lrtfs-real} shows how the methodology for complex signals can be adapted to real-valued signals. \new{Section~\ref{sec:related} discusses how LRTFS relates to other temporal models with low-rank spectrograms and/or structured variance \cite{Kameoka2015,Kameoka2017,Liutkus2011,Yoshii2013,Turner2014,Liutkus2017,Yoshii2018}.} Section~\ref{sec:hybrid} describes how LRTFS can accommodate more advanced multi-layer decompositions in which every layer can have its own t-f resolution or structure (e.g., a sparse instead of low-rank time-frequency structure). 
	Section~\ref{sec:compress} describes a new approach to compressive sampling, which exploits latent low-rank time-frequency structure instead of sparsity, with superior results for the considered type of data. The article is illustrated throughout with experiments on audio signals (the presented methodology is however not limited to audio signals). \new{In particular, we use a running piano toy example to illustrate every stage of our contributions.}

	This article unifies and continues our work presented in the conference papers \cite{nips14,Fevotte2015a}. In particular, it provides more detailed experiments and presents the following novel methodological contributions: the case of real-valued signals (which require to properly handle the Hermitian symmetry of their synthesis coefficients) is now rigorously treated in Section~\ref{sec:lrtfs-real}, algorithm accelerations are presented in Section~\ref{sec:mjle}, and the concept of compressive LRTFS presented in Section~\ref{sec:compress} is entirely novel.

	\section{A baseline: Itakura-Saito NMF and the Gaussian Composite Model (GCM) \label{sec:isnmf}}

	NMF was originally designed in a deterministic setting \cite{lee99}: a measure of fit between $\ve{P}$ and $\W \H$ is minimized with respect to (w.r.t) $\W$ and $\H$. Choosing the ``right'' measure for a specific type of data and task is not straightforward. Furthermore, NMF-based spectral decompositions often arbitrarily discard phase information: only the magnitude of the complex-valued short-time Fourier transform (STFT) is considered. To remedy these limitations, a generative probabilistic latent factor model of the STFT, the GCM, was proposed in \cite{neco09}. It is defined by 
	\bal{ \label{eqn:gcm}
	y_{fn} \sim N_{c}(0, [\W\H]_{fn}),
	}
	where $N_{c}$ refers to the circular complex-valued normal distribution.\footnote{A random variable $x$ has distribution $N_{c}(x|\mu,\lambda) = (\pi \lambda)^{-1} \exp - (|x-\mu|^{2}/\lambda)$ if and only if its real and imaginary parts are independent and with distribution $N( \Re[\mu],\lambda/2)$ and $N(\Im[\mu],\lambda/2)$, respectively.} As shown by Eq.~\eqref{eqn:gcm}, in the GCM the STFT is assumed centered and its variance has a low-rank structure. Many temporal waveforms (such as audio signals) can be assumed centered and this remains true for their Fourier coefficients by linearity of the transformation. This explains the zero-mean assumption in the GCM on the one hand. The low-rank variance structure on the other hand underlies a composite signal structure that makes the model relevant for decomposition task. Indeed, introducing the latent complex-valued components $y_{kfn}$, Eq.~\eqref{eqn:gcm} is equivalent to
	\bal{
	y_{fn} &= \sum_{k} y_{kfn}, \label{eqn:gcm1} \\
	y_{kfn} &\sim N_{c}(0, w_{fk} h_{kn}). \label{eqn:gcm2}
	}
	The latent component $\ve{Y}_{k}$ with coefficients $\{y_{kfn}\}_{fn}$ reflects the contribution of the spectral pattern $\ve{w}_{k}$, the $k^{th}$ column of $\W$, amplitude-modulated in time by the activation coefficients of the $k^{th}$ row of $\H$.

	Under these assumptions, the negative log-likelihood $-\log p(\ve{Y}|\W,\H)$ is equal, up to a constant, to the Itakura-Saito (IS) divergence $D_{\IS}(\ve{P}|\W \H)$ between the power spectrogram $\ve{P} = |\ve{Y}|^{2}$ and $\W \H$. The IS divergence between nonnegative matrices $\ve{A}$ and $\ve{B}$ is defined by 
	\bal{
	D_{\IS}(\ve{A}| \ve{B}) = \sum_{ij} \frac{a_{ij}}{b_{ij}} - \log \frac{a_{ij}}{b_{ij}} -1 .
	}

	The GCM is a step forward from traditional NMF approaches that fail to provide a valid generative model of the STFT itself -- other approaches have only considered probabilistic models of the magnitude spectrogram under Poisson or multinomial assumptions, see \cite{Smaragdis2014} for a review. Still, the GCM is not yet a generative model of the raw signal $x(t)$ itself, but of its STFT. LRTFS fills in this ultimate gap.

	\section{Low-rank time-frequency synthesis (LRTFS) \label{sec:lrtfs}}

	In this section we first present LRTFS for complex-valued signals, closely following \cite{nips14}. Then we rigorously address the case of real-valued signals represented as a complex-valued linear combination of complex-valued t-f atoms (such as Gabor atoms) with Hermitian symmetry. Finally, we discuss relevant connections with the state-of-the-art and illustrate the potential of LRTFS on an audio example.

	\subsection{Complex-valued signals \label{sec:lrtfs-comp}}

	\subsubsection{Model} \label{sec:lrtfs-mod}

	Let ${x}(t)$ denote a complex-valued signal of length $T$ and $\{\phi_{fn}(t)\}_{f=1..F,n=1..N}$ denote a dictionary of complex-valued t-f atoms of length $T$. LRTFS is defined as follows. For $t = 1, \ldots, T$, $f = 1,\ldots,F$, $n = 1,\ldots,N$:
	\begin{align}
		x(t) &= \sum_{fn} \alpha_{fn} \phi_{fn}(t) + e(t), \label{eqn:sig} \\
		\alpha_{fn} &\sim N_{c}(0,[\W \H]_{fn}), \label{eqn:alpha} \\
		e(t) &\sim N_{c}(0,\lambda), \label{eqn:res}
	\end{align}
	where $\{ \alpha_{fn} \}$ are the complex-valued synthesis coefficients, $\W$ and $\H$ are nonnegative matrices of sizes $F \times K$ and $K \times N$, respectively, and $e(t)$ is an additive complex-valued residual term with Gaussian distribution $N_{c}(0,\lambda)$. The synthesis coefficients $\{ \alpha_{fn} \}$ are furthermore assumed independent given $\W$ and $\H$. The synthesis coefficients are dual of the analysis coefficients, defined by $y_{fn} = \sum_{t} x(t) \phi_{fn}^{*}(t)$, where $\cdot^{*}$ denotes conjugation. IS-NMF assumes that the {\it analysis} coefficients follow a GCM, see Eq.~\eqref{eqn:gcm}. In contrast, LRTFS assumes that the {\it synthesis} coefficients follow a GCM, as given by Eq.~\eqref{eqn:alpha}. As announced, LRTFS provides a generative model of the raw data $x(t)$, where IS-NMF only provides a generative model of the transformed data $\ve{Y}$.

Let us denote by $\ve{x}$	and $\ve{e}$ the column vectors of size $T$ with coefficients $x(t)$ and $e(t)$, respectively. \new{Let ${\cal I}$ be an arbitrary one-to-one ``vectorizing'' mapping from $(f,n) \in \{1,\ldots,F\} \times \{1,\ldots,N\} $ to $m = {\cal I}(f,n)  \in \{1,\ldots, M\} $, where $M=FN$. We denote by $\balpha$ the column vector of dimension $M$ with coefficients $\alpha_{m} = \alpha_{ {\cal I}(f,n)} = \alpha_{fn} $. We are abusing the notations by indexing the synthesis coefficients by either $m$ (unstructured vectorized form) or $(f,n)$ (matrix form where $f$ indexes frequencies and $n$ indexes time frames). It should be understood that $m$ and $(f,n)$ are in one-to-one correspondence and the meaning should be clear from the context. The notation $\alpha_{m}$ discards the inherent t-f structure of the coefficients while the notation $\alpha_{fn}$ makes it explicit. Despite abusing, this convention allows to significantly reduce cluttering in the following. Similarly, we denote by $\bPhi$ the matrix of size $T\times M$ with columns $\bphi_{m} = \bphi_{fn}$, where $\bphi_{fn}$ is the column vector of size $T$ with coefficients $\phi_{fn}(t)$. We denote by $\ve{v}$ the column vector of dimension $M$ with coefficients $v_{m} = v_{fn} \defequal [\W \H]_{fn}$. We will sometimes write $\ve{v} = \text{vect}[\W \H]$, where $\text{vect}[\cdot]$ refers to the vectorizing operator induced by ${\cal I}$.}	
%	In the following we will sometimes abuse notations by indexing the coefficients of $\balpha$ or $\ve{v}$ by either $m$ or $(f,n)$, i.e., such that $\alpha_{m} = \alpha_{fn}$ and $v_{m} = v_{fn}$.	
%	It should be understood that $m$ and $(f,n)$ are in one-to-one correspondence, such that  	
%	the notation should be clear from the context. 
%
	Equipped with these notations, we may write Eq.~\eqref{eqn:sig} and~\eqref{eqn:alpha} as
	\bal{
	\ve{x} & = \bPhi \balpha + \ve{e},  \label{eqn:matform} \\
	\balpha &\sim N_{c}(\ve{0}, \text{diag}(\ve{v})), \label{eqn:prioralph} \\
	\ve{e} &\sim N_{c}( \ve{0}, \lambda \, \ve{I}_{T}). \label{eqn:priore}
	}
	Ignoring the low-rank structure of $\ve{v}$, Eqs.~\eqref{eqn:matform}-\eqref{eqn:priore} resemble sparse Bayesian learning (SBL), as introduced in \cite{tip01,Wipf2004}, where it is shown that marginal likelihood estimation of the variance induces sparse solutions of $\ve{v}$ (and as a consequence, of $\balpha$). The essential difference between our model and SBL is that the coefficients are no longer unstructured in LRTFS. Indeed, in SBL, each coefficient $\alpha_{m}$ has a free variance parameter $v_{m}$. This property is fundamental to the sparsity-inducing effect of SBL \cite{tip01}. In contrast, in LRTFS, the variances are now tied together and such that $v_{m} = v_{fn} = [\W \H]_{fn}\ $.

	\subsubsection{Maximum joint likelihood estimation \label{sec:mjle}}

	We now address the estimation of $\W$, $\H$ and $\balpha$ and possibly $\lambda$ in LRTFS. We consider maximum joint likelihood estimation (MJLE), also referred to as type-I maximum likelihood estimation in \cite{Wipf2004}. MJLE relies on the minimization of the following objective function:
	\bal{
	&C_{\JL}(\balpha,\W,\H,\lambda) \defequal - \log p(\ve{x},\balpha|\W,\H,\lambda) \\
	& = -\log p(\ve{x}|\balpha,\lambda) - \log p(\balpha|\W,\H) \\
	& =  \frac{1}{\lambda} \|\ve{x} - \bPhi\balpha\|_2^2 + \sum_{fn} \left[ \frac{|\alpha_{fn}|^2}{[\W\H]_{fn}} + \log{[\W\H]_{fn}} \right] + cst \label{eqn:shrink} \\
	& = \frac{1}{\lambda} \|\ve{x} - \bPhi\balpha\|_2^2 + D_{\text{IS}}(|\balpha|^2|\ve{v}) + \log(|\balpha|^2) + cst  \label{eqn:cjl}
	}
	where $cst = T \log\lambda + (T+M) \log \pi $ and we recall that $\ve{v} = \text{vect}[\W \H]$. %The first term in Eq.~\eqref{eqn:cjl} measures the discrepancy between the raw signal and its approximation. The second term ensures that the synthesis coefficients are approximately low-rank. 

	Another possible estimation procedure for LRTFS is maximum marginal likelihood estimation (MMLE), also referred to as type-II maximum likelihood estimation in \cite{Wipf2004}. It relies on the minimization of $- \log p(\ve{x}|\W,\H,\lambda)$, i.e., involves the marginalization of $\balpha$ from the joint likelihood, following the principle of SBL. We considered MMLE for LRTFS in \cite{nips14} and presented a valid EM algorithm. However our implementation does not scale with the dimensions involved in signal processing, \new{as it requires the estimation of the diagonal elements of the inverse of a $M\times M$ matrix.} Large-scale algorithms for MMLE are left as future work.

	\subsubsection{Alternate minimization algorithm for MJLE \label{sec:algo}}

	We now describe an alternate minimization algorithm that returns stationary points of $C_{\JL}(\btheta)$, where $\btheta = \{\balpha,\W,\H,\lambda\}$. The optimization of $\balpha$ given the other parameters reduces to
	\bal{
	\label{eq:min_alpha}
	\min_{\balpha \in \mathbb{C}^{M}} \frac{1}{\lambda} \|\ve{x} - \bPhi\balpha\|_2^2  + \sum_{fn} \frac{\left|\alpha_{fn}\right|^2}{[\W \H]_{fn}}
	}
	which defines a convex ridge regression problem. The problem has the closed-form solution
	\bal{ \label{eqn:sol}
	\hat{\balpha} = \left[ \bPhi\hermconj \bPhi + \lambda \diag{\ve{v}}^{-1} \right]^{-1} \bPhi\hermconj \ve{x}
	}
	{where $\cdot\hermconj$ denotes conjugate transpose. Eq.~\eqref{eqn:sol} involves the resolution of a linear system of size $M\times M$. The linear system can be reduced to dimension $T$ thanks to the Woodbury identity, but this is still too large in typical signal processing applications. {Computing Eq.~\eqref{eqn:sol} can be done efficiently with} a numerical optimization procedure and several options are available, such as conjugate gradient descent (GCD), expectation-minimization (EM), forward-backward optimization or majorization-minimization. The latter three are closely related and lead in the present case to an iterative shrinkage algorithm (ISA)\cite{Figueiredo2003,Beck2009}. }
	We used in our implementation a complex-valued version of ISA, similar to the complex-valued cases treated in \cite{Chaari2011,Florescu2014}, and using the acceleration described in \cite{chambolle2015convergence}. This leads to a simple and parameter-free implementation with satisfactory speed of convergence. This in particular results in a faster algorithm than the original EM algorithm presented in our initial contribution \cite{nips14}. \new{Moreover, ISA leads in our case to a simpler algorithm than GCD, with same theoretical and practical speed of convergence~\cite{Ben-Tal2004}}. The resulting updates are given in Algorithm~\ref{alg:LRTFS}. The value of the inverse step-size $L$ should be set to the maximum eigenvalue of $\bPhi\hermconj \bPhi$, i.e., the squared spectral norm of $\bPhi$. If this value is not available in closed form or difficult to compute, a larger value $L \ge \| \bPhi \|_{2}^{2}$ is also permissible but will result in smaller step sizes. In Algorithm~\ref{alg:LRTFS}, the operations $\ve{A} \circ \ve{B}$, ${\ve{A}}^{\circ p }$ and $\frac{\ve{A}}{\ve{B}}$ denote entry-wise multiplication, exponentiation and division, respectively.

	The optimization of $\W$ and $\H$ given $\balpha$ reduces to
	\bal{
	\label{eq:min_WH}
	\min_{\W,\H \ge {0}} \ \sum_{fn} D_{\IS}( |\alpha_{fn}|^{2} | [\W \H]_{fn})
	}
	which defines a IS-NMF problem with input matrix $\ve{S} = [|\alpha_{fn}|^{2}]_{fn}$. This a non-convex problem that is generally approached with alternating updates of $\W$ and $\H$ and majorization-minimization (MM) \cite{betanmf}. This results in the multiplicative updates given in Algorithm~\ref{alg:LRTFS}. 

	Finally, the optimization of $\lambda$ given $\balpha$ is trivially given by $\hat{\lambda} = \| \ve{x} - \bPhi \balpha \|/T$. However, the MJLE setting is known to be inefficient for the estimation of both the variance parameters of $\balpha$ and of $\ve{e}$, with either $\bPhi \hat{\balpha}$ or $\hat{\ve{e}}$ capturing most of the signal variance. As such, though the estimation of $\lambda$ is possible in principle, we will consider $\lambda$ to be a fixed hyper-parameter in the following.

	The objective function $C_{\JL}$ being non-convex
	and because we are using an alternate minimization algorithm, 
	the output of Algorithm~\ref{alg:LRTFS} depends on the initialization. In all simulations we initialized the synthesis coefficients $\balpha$ with the \new{analysis} coefficients $\bPhi\hermconj\ve{x}$. The matrices $\W$ and $\H$ are initialized using the absolute values of the complex SVD of the synthesis coefficients \cite{becker2015complex}. Finally, a tempering strategy with warm restart is used to speed up convergence for small target values of $\lambda$. The hyper-parameter $\lambda$ is set to an arbitrarily large value in the first iterations and is then gradually decreased to its target value, as proposed in~\cite{hale2008fixed}. \new{Convergence of the main and inner loops is monitored using the relative difference in norm between successive parameter iterates, see Section~\ref{sec:exlrtfs}.}
	
%	 \new{The convergence is given in practice by some automatic criterion (see Section~\ref{sec:exlrtfs} for the setup).}

	\begin{algorithm}[h!]
		\label{alg:LRTFS}
		Set $L = \| \bPhi \|_{2}^{2}$ (or a larger value)\\		
		Compute the synthesis coefficients $\ve{y} = \bPhi\hermconj\ve{x}$ (with matrix form $\ve{Y}$) \\
		Set $\balpha^{(0)} = \ve{y}$ \\
		Initialize $\W^{(0)}$ and $\H^{(0)}$ with the absolute values of the complex SVD of $\ve{Y}$ \\
		Set $i=0$
	
		\Repeat{convergence}{
		
		\%\% {Update $\W$ and $\H$ with MM}\\
		Compute spectrogram $\ve{S}^{(i)} = \left[|\alpha_{fn}^{(i)}|^{2} \right]_{fn}$ \\
		Initialize inner loop: $\W =\W^{(i)} $, $ \H=\H^{(i)}$ \\
						
		\Repeat{convergence}{

		$\W \leftarrow \W \circ \frac{[\ve{S}^{(i)} \circ (\W \H)^{\circ-2}]\H^\top}{[(\W \H)^{\circ-1}]\H^\top}$ \\
		$\ \H \leftarrow \H \circ \frac{ \W^\top [ \ve{S}^{(i)} \circ (\W \H)^{\circ-2}]}{\W^\top[(\W \H)^{\circ-1}]}$
			
		}
			
		Leave inner loop: $\W^{(i+1)}=\W$, $\H^{(i+1)}=\H$ \\
		Set $\bv^{(i+1)}= \text{vect}[\W^{(i+1)}\H^{(i+1)}]$
		
		\medskip
				
		\%\%  {Update $\balpha$ with accelerated ISA} \\			
		Initialize inner loop: $\ve{a}^{(0)} = \ve{z}^{(0)}=\balpha^{(i)}$ \\
		Set  $j=0$
		
		\Repeat{convergence }{
		\% Descend \\
		$\ve{z}^{(j+1/2)} = \ve{a}^{(j)} + \frac{1}{L} \bPhi\hermconj(\ve{x}-\bPhi\ve{a}^{(j)})$ \\
		\% Shrink \\
		$\ve{z}^{(j+1)} = \frac{\ve{v}^{(i+1)}}{\ve{v}^{(i+1)}+ {\lambda}/{L}} \circ \ve{z}^{(j+1/2)}$ \\
		\% Accelerate \\
		$\ve{a}^{(j+1)} = \ve{z}^{(j+1)}+\frac{j}{j+5}(\ve{z}^{(j+1)}-\ve{z}^{(j)})$ \\
		{ $j \leftarrow j+1$}
		}
		Leave inner loop: $\balpha^{(i+1)} = \ve{z}^{(j+1)}$ 
		}
		
		%\new{Output: $\hat{\balpha}$, $\hat{\W}$, $\hat{\H}$}
		\caption{Alternate minimization for LRTFS}
	\end{algorithm}

	\subsubsection{Reconstruction of the latent components}

\new{Algorithm~\ref{alg:LRTFS} outputs an estimate of $\balpha$, $\W$, $\H$. The approximate signal can directly be recovered from the estimated synthesis coefficients as $\hat{\ve{x}} = \bPhi \hat{\balpha}$.} LRTFS further assumes that the synthesis coefficients follow a GCM, see Eq.~\eqref{eqn:alpha}. As such, $\alpha_{fn}$ may be written as a sum of Gaussian latent components, such that $\alpha_{fn} = \sum_{k} \alpha_{kfn}$, with $\alpha_{kfn} \sim N_{c}(0,w_{fk} h_{kn})$. Denoting by $\balpha_{k}$ the column vector of dimension $M$ with coefficients $\{\alpha_{kfn} \}_{fn}$, Eq.~\eqref{eqn:matform} may be written as
	\bal{ \label{eqn:comp}
	\ve{x} = \sum_{k} \bPhi \balpha_{k} + \ve{e} = \sum_k {\ve{c}_{k}} + \ve{e}\ ,
	}
	where $\ve{c}_{k} = \bPhi \balpha_{k} $. The component $\ve{c}_{k}$ is the ``temporal expression'' of spectral pattern $\ve{w}_{k}$, the $k^{th}$ column of $\ve{W}$. Given estimates of $\balpha$, $\W$ and $\H$, the components may be reconstructed \new{a posteriori} in various ways. A natural choice is $\hat{\ve{c}}_{k}^{\text{MMSE}}= \bPhi \hat{\balpha}_{k}^{\text{MMSE}}$ with 
	\bal{
	\hat{\balpha}_{k}^{\text{MMSE}} \, \defequal \, \EE [\balpha_{k} | \ve{x}, \hat{\btheta} ] = \EE[\balpha_{k} | \hat{\balpha}, \hat{\W}, \hat{\H} ].
	}
	The coefficients of $\hat{\balpha}_{k}^{\text{MMSE}}$ are given by
	\bal{ \label{eqn:mmse}
	\hat{\alpha}_{kfn}^{\text{MMSE}} =  \frac{\hat{w}_{fk} \hat{h}_{kn}}{\ [\hat{\W} \hat{\H}]_{fn}} \hat{\alpha}_{fn}.
	}
	Using this estimate, the latent components are reconstructed by applying a t-f dependent ``Wiener mask'' to the synthesis coefficients. This procedure and the expression of $\hat{\alpha}_{fkn}^{\text{MMSE}}$ is analog to the standard Wiener estimate of the latent components in IS-NMF applied to $| \ve{Y}|^{2}$ \cite{neco09} and 
	given by
	\bal{ \label{eqn:wienmf}
	\hat{y}_{kfn}^{\text{MMSE}} = \frac{\hat{w}_{fk} \hat{h}_{kn}}{\ [\hat{\W} \hat{\H}]_{fn}} y_{fn}.
	}
	The estimate $\hat{\balpha}$ is used as an intermediate variable in the expression of $\hat{\alpha}_{kfn}^{\text{MMSE}}$ given by Eq.~\eqref{eqn:mmse}. Another possible estimate, which marginalizes $\balpha$, is
	\bal{
	\hat{\balpha}_{k} \, = \, \EE [ \balpha_{k} | \ve{x}, \hat{\W}, \hat{\H}, \hat{\lambda} ],
%	& = \diag{\hat{\ve{v}}_{k}} \bPhi\hermconj [\bPhi \diag{\hat{\ve{v}}} \bPhi\hermconj + \lambda \ve{I}_{T} ]^{-1} \ve{x},
	}
	where $\hat{\ve{v}}_{k}$ is the vector of dimension $M$ with coefficients $\{\hat{w}_{fk}\hat{h}_{kn}\}_{fn}$. The input of the estimator is now the raw data $\ve{x}$ which may be more sensible. \new{The expression of this alternative estimate can be derived in closed-form but the resulting expression involves the large-scale inversion of $T \times T$ matrices, which is hardly feasible in practice.}

%	 \modif{It however requires the resolution of a large-scale linear system of dimension $\min \{T,M\}$, which again is hardly feasible in usual signal processing scenarios. CHECK}

	\subsection{Real-valued signals \label{sec:lrtfs-real}}

	\subsubsection{Model} \label{sec:rLRTFS}

	In many signal processing settings the data is a real-valued signal $x(t)$ expressed as a linear combination of complex-valued t-f atoms with Hermitian symmetry. More specifically, the dictionary and synthesis coefficients are such that $\bphi_{fn} = \bphi_{(F-f)n}^{*}$ and $\alpha_{fn} = \alpha_{(F-f)n}^{*}$ for $f=1,\ldots, F/2$ (assuming $F$ to be even-valued for simplicity), where $\cdot^{*}$ denotes conjugation. Under this particular structure, we have
	\bal{
	\sum_{f=1}^{F} \sum_{n=1}^{N} \alpha_{fn} \phi_{fn}(t) = \sum_{f=1}^{F/2} \sum_{n=1}^{N} 2 \Re[\alpha_{fn} \phi_{fn}(t)]
	}
	and we define real-valued LRTFS (rLFTS) as follows. For $t = 1, \ldots, T$, $f = 1,\ldots,F/2$, $n = 1,\ldots,N$:
	\begin{align}
		x(t) &= \sum_{f=1}^{F/2} \sum_{n=1}^{N} 2 \Re[\alpha_{fn} \phi_{fn}(t)] + e(t) \label{eqn:sigR}, \\
		\alpha_{fn} &\sim N_{c}(0,[\W \H]_{fn}) \label{eqn:alphaR}, \\
		e(t) &\sim N(0,\lambda) \label{eqn:resR}.
	\end{align}
	Note how $F$ now runs from $1$ to $F/2$ instead of 1 to $F$. The synthesis coefficients $\alpha_{fn}$ remain complex-valued and the residual $e(t)$ becomes real-valued. $\W$ and $\H$ are nonnegative matrices of sizes $F/2 \times K$ and $K \times N$, respectively.

	Let us now denote by $\underline{\balpha}$ and ${\ve{v}}$ the vectors of dimension $M/2$ with coefficients $\alpha_{fn}$ and $v_{fn} = [\W \H]_{fn}$, respectively, and by $\underline{\bPhi}$ the matrix of dimension $T \times M/2$ with columns $\bphi_{fn}$, for $f=1,\ldots,F/2$ and $n=1,\ldots,N$. With these notations we have $\balpha = [\underline{\balpha}^\top, \underline{\balpha}\hermconj]^\top$, $\bPhi = [\underline{\bPhi}, \underline{\bPhi}^{*}]$ and $\bPhi \balpha = 2\Re[\underline{\bPhi} \underline{\balpha}]$. Consequently, we may write Eq.~\eqref{eqn:sigR}-\eqref{eqn:resR} as
	\bal{
	\ve{x} & = 2 \Re[\underline{\bPhi} \underline{\balpha}] + \ve{e}  \label{eqn:matformR}, \\
	\underline{\balpha} &\sim N_{c}(\ve{0}, \text{diag}({\ve{v}})) \label{eqn:prioralphR}, \\
	\ve{e} &\sim N( \ve{0}, \lambda \, \ve{I}_{T}) \label{eqn:prioreR}.
	}

	\subsubsection{Estimation} \label{sec:rLFTSestim}

	The MJLE objective function for rLRTFS writes
	\bal{
	&C_{\JL}^{\Re}(\ul{\balpha},\W,\H,\lambda) \defequal - \log p(\ve{x},\bualpha|\W,\H,\lambda) \\
	& =  \frac{1}{2 \lambda} \|\ve{x} - 2\Re[\buPhi \bualpha]\|_2^2 \\
	& \quad + \sum_{f=1}^{F/2} \sum_{n=1}^{N} \left[ \frac{|\alpha_{fn}|^2}{[\W\H]_{fn}} + \log{[\W\H]_{fn}} \right] + cst \\
	& = \frac{1}{2 \lambda} \|\ve{x} - 2\Re[\buPhi\bualpha]\|_2^2 + D_{\text{IS}}(|\bualpha|^2|{\ve{v}}) + \log(|\bualpha|^2) + cst  \label{eqn:cjl2}
	}
	where $cst = \frac{T}{2} \log (2\pi\lambda) + \frac{M}{2} \log \pi$. Using an alternate minimization setting like in Section~\ref{sec:algo}, the updates of $\W$ and $\H$ are virtually unchanged. They amount to IS-NMF of the matrix form of the synthesis spectrogram $|\bualpha|^{2}$ (of size $F/2 \times N$). The update of $\lambda$ is easily given by $\hat{\lambda} = \|\ve{x} - 2\Re[\buPhi \bualpha]\|_2^2 /T$, but here again we prefer to treat $\lambda$ as an hyper-parameter. The update of $\bualpha$ involves the following minimization problem:
	\bal{
	\min_{\bualpha \in \mathbb{C}^{M/2}} F(\bualpha) \defequal \frac{1}{2\lambda} \|\ve{x} - \Re[\buPhi\bualpha] \|_2^2  + \sum_{f=1}^{F/2} \sum_{n=1}^{N} \frac{\left|\alpha_{fn}\right|^2}{[\W \H]_{fn}}.  \label{eqn:minar}
	}
	The problem defined by Eq.~\eqref{eqn:minar} has a closed-form solution, with a less simpler expression than Eq.~\eqref{eqn:sol}. The solution is still computationally demanding and the following numerical procedure is preferable. 

	\begin{theorem}[Iterative shrinking algorithm for rLRTFS] 
		\label{thm}
		Let $L = \| \bPhi \|_{2}^{2}$ (with $\bPhi = [\underline{\bPhi}, \underline{\bPhi}^{*}]$) and $\bualpha^{(0)}$ be an initial estimate. The following sequence of updates converge to the global solution of problem~\eqref{eqn:minar}:
		\bal{
		\bualpha^{(j+1/2)} &= \bualpha^{(j)} + \frac{1}{L} \buPhi\hermconj (\ve{x} - 2 \Re[\buPhi \bualpha^{(j)}]), \label{eqn:alph1} \\
		\bualpha^{(j)} & = \frac{\ve{v}}{\ve{v} +  \lambda/L} \circ \bualpha^{(j+1/2)}. \label{eqn:alph2}
		}
	\end{theorem}

\begin{proof}
		The proof consists in reformulating~Eq.~\eqref{eqn:minar} as a quadratic optimization problem over the real and imaginary parts of $\bualpha$ and applying ISA. 
		Let $\ve{A} = 2 [\Re[\buPhi], -\Im[\buPhi]] $, $\ve{b} = [\Re[\bualpha]^\top \Im[\bualpha]^\top]^\top$ and $\ve{c} = \frac{1}{2} [\ve{v}^\top \ve{v}^\top]^\top$. Then we may write
		\bal{ \label{eqn:minproof}
		F(\bualpha) = F(\ve{b}) = \frac{1}{2\lambda} \| \ve{x} - \ve{A} \ve{b} \|_{2}^{2} + \frac{1}{2} \sum_{m=1}^{M} \frac{b_{m}^{2}}{c_{m}}.
		}
		Denoting $L_{\ve{A}} = \| \ve{A} \|_{2}^{2}$, the ISA update for problem~\eqref{eqn:minproof} writes \cite{Figueiredo2003,Beck2009}
		\bal{
		\ve{b}^{(j+1/2)} &= \ve{b}^{(j)} + \frac{1}{L_{\ve{A}}} \ve{A}^\top (\ve{x} - \ve{A} \ve{b}^{(j)}), \label{eqn:b1} \\
		\ve{b}^{(j)} & = \frac{\ve{c}}{\ve{c} + \lambda/L_{\ve{A}}} \circ \ve{b}^{(j+1/2)}. \label{eqn:b2}
		}
		Using the identities $\ve{A} \ve{b} =  2 \Re[\buPhi \bualpha]$ and
		\bal{
		\ve{A}^\top \ve{e} = 
		2 \begin{bmatrix} \Re[\buPhi\hermconj \ve{e}] \\ \Im[\buPhi\hermconj \ve{e}] \end{bmatrix},
		}
		Eqs.~\eqref{eqn:b1}-\eqref{eqn:b2} can be rearranged in complex form as
		\bal{
		\bualpha^{(j+1/2)} &= \bualpha^{(j)} + \frac{2}{L_{\ve{A}}} \buPhi\hermconj (\ve{x} - 2 \Re[\buPhi \bualpha^{(j)}]) \\
		\bualpha^{(j)} & = \frac{\ve{v}}{\ve{v} +  2 \lambda/L_{\ve{A}}} \circ \bualpha^{(j+1/2)}.
		}
		To complete the proof we only need to show that $L =  2 L_{\ve{A}} $. Let $\balpha = [\bualpha_{1}^\top, \bualpha_{2}^\top]^\top$ be an eigenvector of $\bPhi\hermconj \bPhi$ with maximum eigenvalue $L = \|\bPhi \|_{2}^{2}$. By definition we have
		\bal{
		\begin{bmatrix}
			\buPhi\hermconj \buPhi & \buPhi\hermconj \buPhi^{*} \\
			\buPhi^\top \buPhi & \buPhi^\top \buPhi^{*}
		\end{bmatrix}
		\begin{bmatrix}
			\bualpha_{1} \\
			\bualpha_{1}
		\end{bmatrix}
		=
		L
		\begin{bmatrix}
			\bualpha_{1} \\
			\bualpha_{1}
		\end{bmatrix}
		\label{eqn:eig}
.}
		By taking the conjugate of Eq.~\eqref{eqn:eig}, we easily show that $[\bualpha_{2}\hermconj, \bualpha_{1}\hermconj]^\top$ is also an eigenvector with eigenvalue $L$. It follows that $[\bualpha_{1} + \bualpha_{2}\hermconj, \bualpha_{2} + \bualpha_{1}\hermconj]^\top$ is also an eigenvector, which happens to have a Hermitian structure. We may thus impose $\bualpha_{2} = \bualpha_{1}^{*}$ and as such $\balpha = [\bualpha^\top, \bualpha\hermconj]^\top$. Then we have the following series of equivalences:
		\bal{
		\bPhi\hermconj \bPhi \balpha = L \balpha &\iff \buPhi\hermconj \buPhi \bualpha + \buPhi\hermconj \buPhi^{*} \bualpha^{*} = L \bualpha \\
		& \iff 2 \buPhi\hermconj \Re[\buPhi \bualpha] = L \bualpha \\
		& \iff \frac{1}{2} \ve{A}^\top \ve{A} \ve{b} = L \ve{b}.
		}
		As such, the spectra of $\buPhi\hermconj \buPhi$ and $\ve{A}^\top \ve{A}$ coincide up to a factor 2 and we have $L = 2 L_{\ve{A}}$, which concludes the proof.
\end{proof}

	\subsubsection{Comments about implementation}

	\begin{algorithm}[!t]
		\label{alg:rLRTFS}
		Set $L = \| \bPhi \|_{2}^{2}$ (or a larger value) \\
		Compute the synthesis coefficients ${\ve{y}} = \buPhi\hermconj\ve{x}$ (with matrix form ${\ve{Y}}$) \\
		Set $\bualpha^{(0)} = {\ve{y}}$ \\
		Initialize $\W^{(0)}$ and $\H^{(0)}$ with the absolute values of the complex SVD of ${\ve{Y}}$ \\
		Set $i=0$
	
		\Repeat{convergence}{
		
		\%\% Update $\W$ and $\H$ with MM \\
		Compute spectrogram $\ve{S}^{(i)} = \left[|\underline{\alpha}_{fn}^{(i)}|^{2} \right]_{f=1,\ldots,F/2, n=1,\ldots,N}$ \\
		Initialize inner loop: $\W =\W^{(i)} $, $ \H=\H^{(i)}$ \\
						
		\Repeat{convergence}{

		$\W \leftarrow \W \circ \frac{[\ve{S}^{(i)} \circ (\W \H)^{\circ-2}]\H^\top}{[(\W \H)^{\circ-1}]\H^\top}$ \\
		$\ \H \leftarrow \H \circ \frac{ \W^\top [ \ve{S}^{(i)} \circ (\W \H)^{\circ-2}]}{\W^\top[(\W \H)^{\circ-1}]}$
			
		}
			
		Leave inner loop: $\W^{(i+1)}=\W$, $\H^{(i+1)}=\H$ \\
		Set $\bv^{(i+1)}= \text{vect}[\W^{(i+1)}\H^{(i+1)}]$
		
		\medskip
				
		\%\% Update $\balpha$ with accelerated ISA \\			
		Initialize inner loop: $\ve{a}^{(0)} = \ve{z}^{(0)}=\bualpha^{(i)}$ \\
		Set  $j=0$
		
		\Repeat{convergence }{
		\% Descend \\
		$\ve{z}^{(j+1/2)} = \ve{a}^{(j)} + \frac{1}{L} \buPhi\hermconj(\ve{x}- 2\Re[\buPhi\ve{a}^{(j)}])$ \\
		\% Shrink \\
		$\ve{z}^{(j+1)} = \frac{\ve{v}^{(i+1)}}{\ve{v}^{(i+1)}+ {\lambda}/{L}} \circ \ve{z}^{(j+1/2)}$ \\
		\% Accelerate \\
		$\ve{a}^{(j+1)} = \ve{z}^{(j+1)}+\frac{j}{j+5}(\ve{z}^{(j+1)}-\ve{z}^{(j)})$ \\
		{ $j \leftarrow j+1$}
		}
		Leave inner loop: $\bualpha^{(i+1)} = \ve{z}^{(j+1)}$ %\ced{ou $\ve{z}^{(j+1)}$ ?}
		
		}
		\caption{Alternate minimization for rLRTFS}
	\end{algorithm}

	Eq.~\eqref{eqn:alph1} and \eqref{eqn:alph2} can be accelerated like before and this results in the general procedure summarized in Algorithm~\ref{alg:rLRTFS}. As compared to Algorithm~\ref{alg:LRTFS}, $\balpha$ is essentially replaced by $\bualpha$, of size half, and the expression of $\ve{z}^{(i+1)}$ is changed with Eq.~\eqref{eqn:alph1}. Although the same notations are used for convenience, $\ve{Y}$, $\ve{S}^{(i)}$ and $\ve{W}$ become matrices with $F/2$ rows.

	Eq.~\eqref{eqn:alph1} can be read as follows. The operation $2 \Re[\buPhi \bualpha]$ consists of reconstructing an approximation $\hat{\ve{x}}$ of $\ve{x}$ based on the current synthesis coefficients $\bualpha$. The operation $\buPhi\hermconj \ve{e}$ then consists in computing the analysis coefficients (restricted to ``positive'' frequencies, i.e., $f=1,\ldots,F/2$) of the current residual $\ve{e} = \ve{x} - \hat{\ve{x}}$. When $\bPhi$ is a tight Gabor frame, these operations can be efficiently performed with dedicated time-frequency libraries, such as the MATLAB \& Python Large Time-Frequency Analysis Toolbox (LTFAT) \cite{Pruuvsa2014}.\footnote{The specific commands being of the like $\bualpha = \text{\tt dgtreal}(\ve{x},\text{`parameters'})$ and $\ve{x} = \text{\tt idgtreal}(\bualpha,\text{`parameters'})$, where $\texttt{dgt}$ stands for discrete Gabor transform.} When the Gabor frame is tight, i.e., $\bPhi \bPhi\hermconj \ve{x} = \ve{x}$, $\bPhi$ has a unit spectral norm and we may set $L=1$. A MATLAB implementation of Algorithm~\ref{alg:rLRTFS} is available online.\footnote{{\url{https://www.irit.fr/~Cedric.Fevotte/extras/tsp2018/}}. Future references to online material refer to this same url.}

	Finally, given estimates of $\bualpha$, $\W$ and $\H$, latent component coefficients $\hat{\bualpha}_{k}$ may be reconstructed like in Eq.~\eqref{eqn:mmse}, and then $\hat{\ve{c}}_{k}=  2 \Re[\buPhi \hat{\bualpha}_{k}]$. 

	%\new{The same complex SVD based initialization as in the complex case can be used to initialize $\W$ and $\H$ {\bf C'est ecrit dans l'algo}.} %\href{{https://www.irit.fr/~Cedric.Fevotte/extras/tsp2018/}}{https://www.irit.fr/$\sim$Cedric.Fevotte/extras/tsp2018/}
	
	\subsection{Related work} \label{sec:related}
	
	The closest to our work are probably the recent papers by Kameoka \cite{Kameoka2015,Kameoka2017} which addresses temporal models of the form $\ve{x} = \sum_{k} \ve{c}_{k}$, like Eq.~\eqref{eqn:comp}, where the spectrograms of the latent components are approximately rank-one. In essence (and slightly simplifying) these papers address optimization problems of the form
	\bal{
	\min_{\ve{c}_{k}, \W, \H }  \sum_{fkn} D([ |\bPhi\hermconj \ve{c}_{k}|^{2}]_{fn} |  w_{fk} h_{kn}) \ \ \text{s.t.} \ \ \ve{x} = \sum_{k} \ve{c}_{k}
	}
	where $|\bPhi\hermconj \ve{c}_{k}|^{2}$ is the spectrogram of $\ve{c}_{k}$, indexed by $f$ and $n$ according to the convention of Section~\ref{sec:lrtfs-mod}, and $D(\cdot | \cdot)$ is a divergence between nonnegative matrices (either the generalized Kullback-Leibler divergence or the quadratic cost in  \cite{Kameoka2015,Kameoka2017}). Though very elegant in our opinion, the approaches of \cite{Kameoka2015,Kameoka2017} are still analysis-based and do not yet provide a fully generative synthesis-based model like LRTFS. 
	
	%Studies of synthesis vs analysis representations can be found in the sparse signal approximation literature~\cite{elad2007analysis,nam2013cosparse}.

	\new{Another related trend of work are the approaches of \cite{Liutkus2011,Yoshii2013,Turner2014} which essentially model $x(t)$ as a sum of variance-structured Gaussian processes. Using our notations, the model in \cite{Turner2014} sets $N=T$ and assumes $x(t) = \sum_{f} \sqrt{[\W \H]_{ft}} \, \Re[s_{f}(t) e^{-j 2\pi ft/F}]$ where $s_{f}(t)$ is a complex-valued Gaussian autoregressive sequence. In \cite{Liutkus2011,Yoshii2013}, $x(t)$ is modeled as a short-time stationary process. The signal is segmented into overlapping temporal frames $\ve{x}_{n}$ of size $P$, like in the first stage of a STFT. Each temporal frame $\ve{x}_{n}$ is then assumed to follow a multivariate Gaussian distribution with covariance $\ve{R}_{n} = \sum_{k} h_{kn}  \bSigma_{k}$, where $\bSigma_{k}$ is a full covariance matrix of size $P \times P$ (real-valued and symmetric). Estimation then consists in estimating the set of parameters $\{ \bSigma_{k} \}_{k}$ and $\H$ from the entire set of temporal frames (the whole approach is coined PSDTF in \cite{Yoshii2013}). When it is further assumed that $\bSigma_{k}$ is the covariance matrix of a real-valued latent stationary process, $\bSigma_{k}$ becomes circulant and is diagonalized by the discrete Fourier transform (DFT). In that case, it can be shown that PSDTF specializes to IS-NMF \cite{Liutkus2011,Yoshii2013}. PSDTF remains close to the analysis view of IS-NMF: the raw data is segmented into overlapping frames and each frame is individually assigned a covariance model. In contrast, LRTFS provides a generative model of the entire signal $x(t)$, assumed to be a linear combination of elementary t-f bricks endowed with a low-rank variance structure. In LRTFS, we assumed the synthesis coefficients to be conditionally mutually independent, i.e., $p(\balpha | \W, \H) = \prod_{fn}(\alpha_{fn} | [\W \H]_{fn})$. We could very well consider more sophisticated models similar to PSDTF or extensions \cite{Liutkus2017,Yoshii2018} which assume some correlation across frequencies or frames. Papers \cite{Liutkus2011,Yoshii2013,Turner2014} describe NMF-related generative probabilistic models rooted in time series analysis while LRTFS offers a different perspective, rooted in the sparse approximation literature. In particular, LRTFS can be used with any time-frequency dictionary $\bPhi$, can easily accommodate multi-layer variants (see Section~\ref{sec:hybrid}) or be considered for inverse problems (see Section~\ref{sec:compress}).}

%	The set of eigenvalues define the Power Spectral Density (PSD). Because $\ve{x}_{n}$ is real-valued, the PSD is symmetric and we denote by $\ve{w}_{k} \in \mathbb{R}_{+}^{\frac{L}{2}}$ the positive frequency part. Denoting by $\buPhi^{\text{DFT}}$ the positive frequency part of the DFT matrix, we may show that under PSDTF, the temporal frames follow the following generative model:
%	\bal{
%	\ve{x}_{n} & = 2 \sum_{f=1}^{L/2}  \Re[ \alpha_{fn} \underline{\bphi}_{f}^{\text{DFT}} ] \\
%		\alpha_{fn}	& \sim {\cal N}_{c}(0, [\W \H]_{fn})
%	}
%This model is reminiscent of Eqs~\eqref{eqn:sigR}-\eqref{eqn:alphaR}. It is equal when the frames $\ve{x}_{n}$ are non-overlapping, un-windowed, and when un-windowed Fourier atoms of temporal support $L$ are used in Eq.~\eqref{eqn:sigR}. 

%Except in this very specific case and though defined in the temporal domain, 

%Each temporal frame $n$ of $\ve{x}$ is assumed to follow a Gaussian process with covariance $\sum_{k} \bSigma_{kn}$. 
	
%Based on the Whittle approximation, $\bSigma_{kn}$ is assumed diagonal in the Fourier basis and its eigenspectrum is approximated by $h_{kn} \ve{w}_{k}$. 

	\subsection{Example \label{sec:exlrtfs}}
	\label{sec:piano1}

	\begin{figure}[t]
		\centering
		(a) Input data \\
		\includegraphics[width=\linewidth,trim={1cm 1cm 0 0},clip]{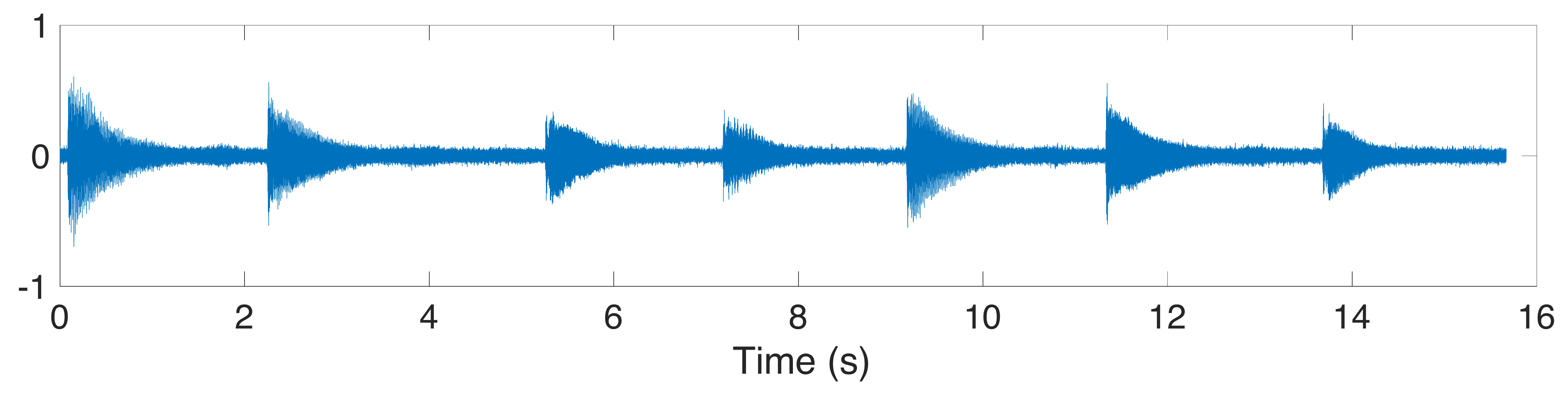} \\
		(b) IS-NMF decomposition
		\includegraphics[width=\linewidth,trim={4cm 2cm 4cm 1cm},clip]{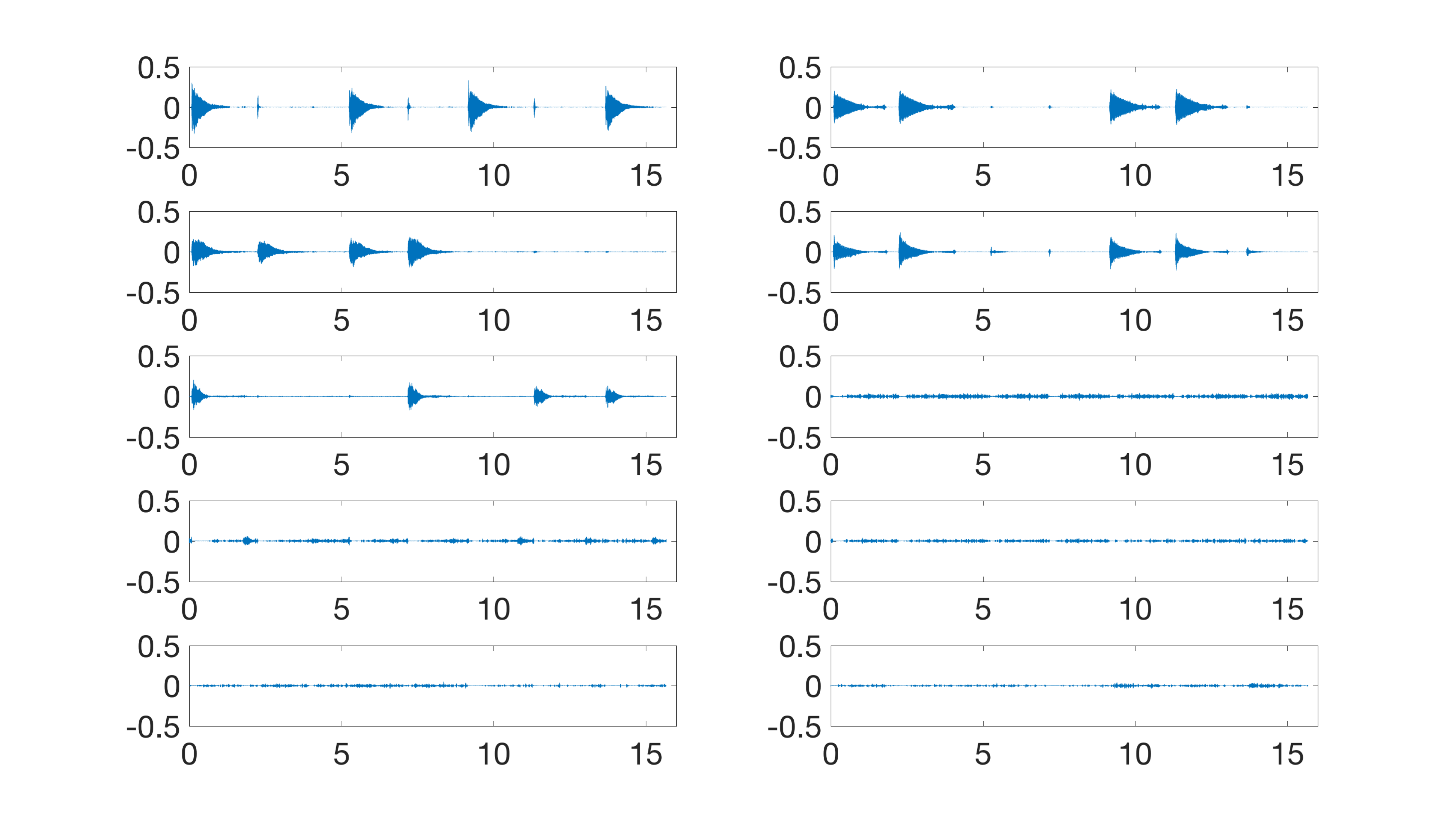}	
		(c) LRTFS decomposition \\
		\includegraphics[width=\linewidth,trim={4cm 2cm 4cm 1cm},clip]{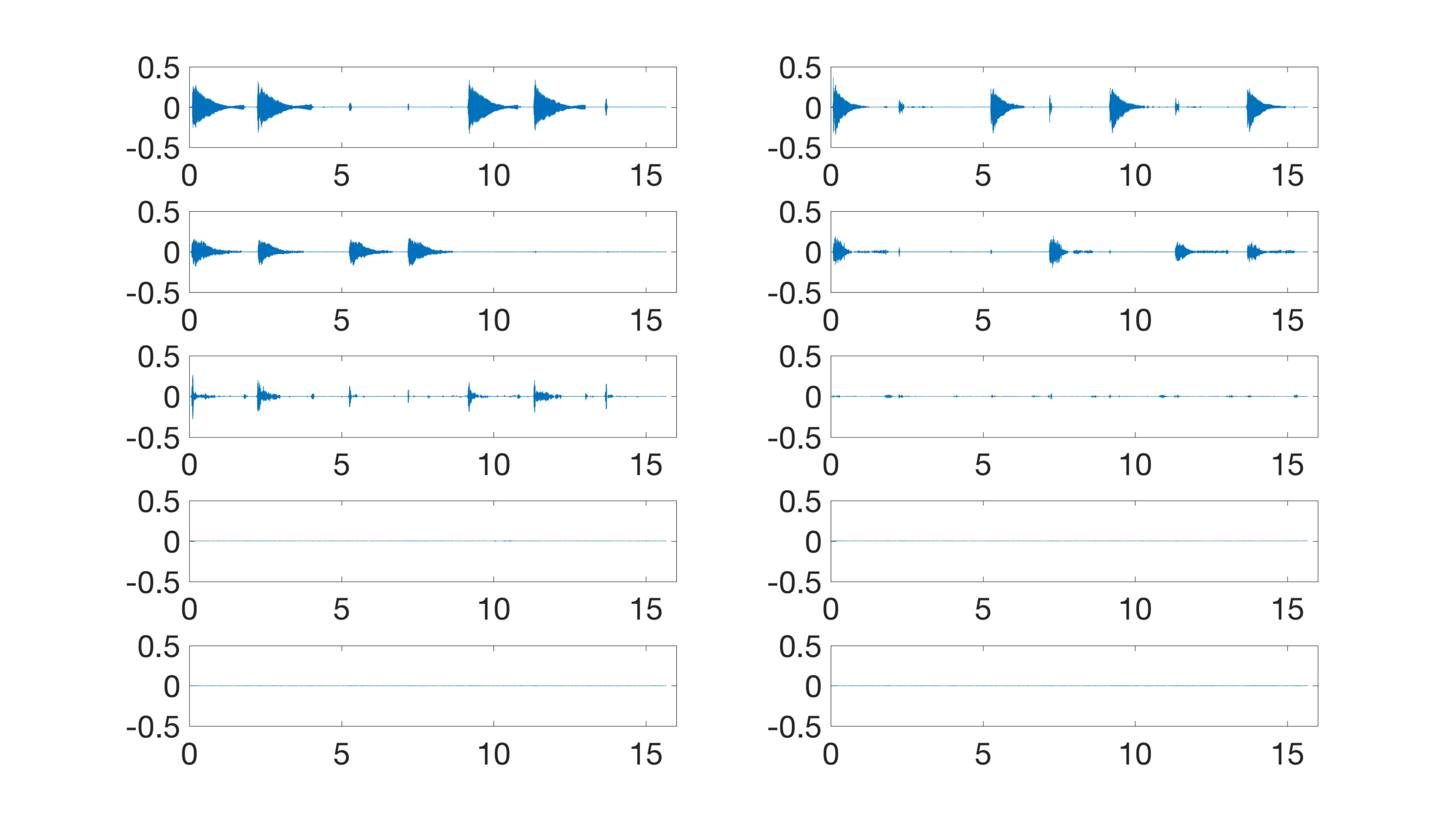} \\
		\caption{Decomposition of a piano sequence consisting of four notes. \new{The subplots in (b) display the latent components obtained by STFT inversion of Eq.~\eqref{eqn:wienmf}. The subplots in (c) display the latent components $\hat{\ve{c}}_{k}^{\text{MMSE}}$.} The components are displayed by decreasing energy (from left to right and top to bottom).}
		\label{fig:piano_1Layer}
	\end{figure}

	\begin{figure}[t]
		\centering
		(a) Analysis coefficients of input data \\
		\includegraphics[width=\linewidth,trim={1cm 1cm 0 0},clip]{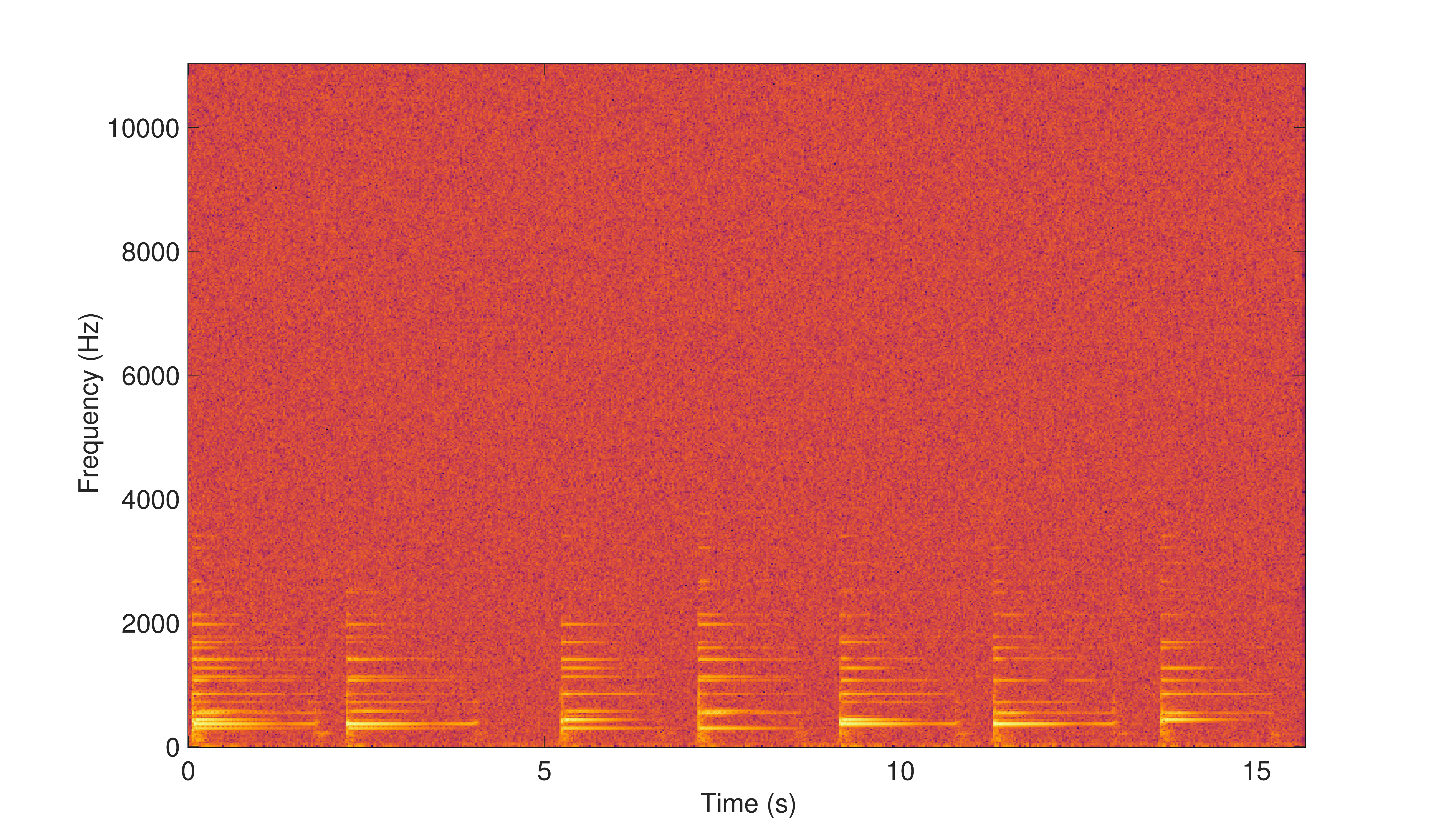} \\
		(b) Synthesis coefficients estimated by LRTFS
		\includegraphics[width=\linewidth,trim={1cm 1cm 0 0},clip]{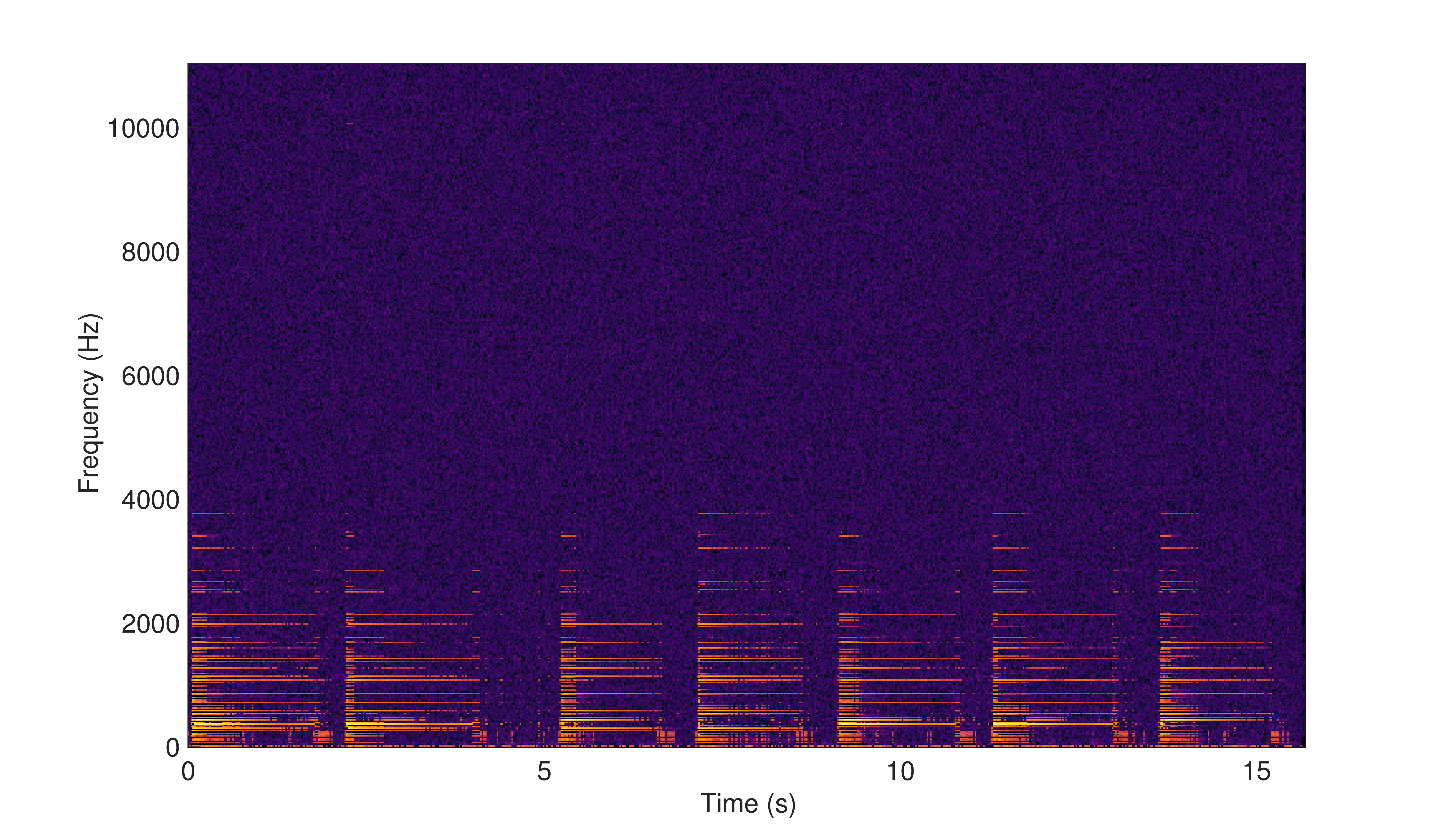}	
%		(c) $WH$ estimated by LRTFS \\
%		\includegraphics[width=\linewidth,trim={1cm 1cm 0 0},clip]{fig/WH_LRTFS} \\
		\caption{\new{Time-frequency analyses of the piano sequence displayed in Fig.~\ref{fig:piano_1Layer}~(a). Subplot (a) displays the squared magnitude of the analysis coefficients given by the STFT, i.e., the power spectrogram $|y_{fn}|^{2}$. Subplot (b) displays the squared magnitude of synthesis coefficients $\alpha$ estimated by LRTFS, i.e., $|\alpha_{fn}|^{2}$. Common dB scale is used on the two subplots. %The subplot in (c) display the low rank matrix $WH$ estimated by LRTFS.
		}}
		\label{fig:piano_1Layer_TF}
	\end{figure}

	We illustrate the performance of LRTFS compared to standard IS-NMF using the piano example used in~\cite{neco09}. The sequence has a simple structure: four notes are played together at once in the first measure and are then played by pairs in all possible combinations in the subsequent measures. The duration is $15.6$~s and the sampling rate $22050$~Hz. In noise-free conditions and with appropriate initialization, \new{standard IS-NMF is able to extract six identifiable latent components from this musical signal: the temporal expression of the four individual notes in a set of four components, the transient parts produced by the hammer hitting the strings {in a fifth component} and the sound produced by the sustain pedal when it is released {in a sixth component}~\cite{neco09}}. We here consider a noisy example using additive white Gaussian noise with $20$~dB input Signal to Noise Ratio (SNR). \new{The resulting signal is displayed in Fig.~\ref{fig:piano_1Layer}~(a).} A tight Gabor dictionary (with Hermitian symmetry) built on a Hann window of $1024$~samples ($46$ ms) with $50\%$ overlap is used for $\bPhi$. IS-NMF is applied to the analysis power spectrogram $|\buPhi \ve{x}|^{2}$ \new{displayed in Fig.~\ref{fig:piano_1Layer_TF}~(a).} The number of latent components is arbitrarily set to $K=10$ for both IS-NMF and rLTFS and the two methods are run from the same initialization (based on the SVD of $\ve{Y}$, see Algorithm~\ref{alg:rLRTFS}). Iteration of the main and inner loops is stopped when the relative error between two successive parameter iterates falls under $10^{-5}$. rLTFS is run with $30$ different values of $\lambda$ logarithmically equally spaced between $10^{-1}$ and $10^{-6}$. \new{The initialization described in Algorithm~\ref{alg:rLRTFS} was used for the first value $\lambda = 10^{-1}$ and warm restart was used for the subsequent experiments. The CPU time for the total 30 experiments is $\sim$5 min using a MATLAB implementation running on a Intel Core i5 processor.} We show results corresponding to the value of $\lambda$ that maximizes the output SNR given by
	\bal{
	10 \log \frac{\| \hat{\ve{x}} - \ve{x} \|_{2}^{2}}{\| \ve{x} \|_{2}^{2}}.
	}
%	\ced{Running time LRTFS: 443 s.}

%\new{The all $30$ run takes $372$~s under MATLAB on Intel Core i5.}
	
	Decomposition results are reported in Fig.~\ref{fig:piano_1Layer}. LRTFS is able to recover the four notes in the first four components, \new{like standard IS-NMF in the noise free case,} while the fifth component recovers the transient components produced by the hammer and the sustain pedal, \new{corresponding to the last two components estimated by standard IS-NMF in the noise free case. As expected,} the remaining five components are inaudible because of the denoising performed by LRTFS. In this \new{noisy} setting, IS-NMF fails to recover this transient part and splits the first note into two components. The input noise is spread over the five remaining components. Audio files are available online. \new{Fig.~\ref{fig:piano_1Layer_TF}~(b) displays the squared magnitude of the synthesis coefficients estimated by LRTFS. Denoising is clearly illustrated by the recovery of high frequencies.}

%\new{The temporal samples of the piano sequence is displayed on Fig~\ref{fig:piano_1Layer}~(a). The temporal latent components estimated by IS-NMF are shown in Fig~\ref{fig:piano_1Layer}~(b), while the temporal latent components estimated by LRTFS are shown in Fig~\ref{fig:piano_1Layer}~(c).}

%\new{In Fig.~\ref{fig:piano_1Layer_TF} we show the analysis coefficients of the STFT of the piano example as well as the synthesis coefficients estimated by LRTFS. The synthesis coefficients are sparse, and the denoising effect is clearly illustrated by the recovery of high frequencies. It is interesting to note that the horizontal structure of the coefficients is well preserved.}

	\section{Multi-layer LRTFS \label{sec:hybrid}}

	Besides the advantage of modeling the raw signal itself, and not its STFT, another major strength of LRTFS is that it offers the possibility of multi-layer modeling. This means we may envisage models of the form
	\bal{ \label{eqn:multil}
	\ve{x} &= \ve{x}_{\text{a}} + \ve{x}_{\text{b}} + \ve{e} = \bPhi_{\text{a}} \balpha_{\text{a}} + \bPhi_{\text{b}} \balpha_{\text{b}} + \ve{e}
	}
	where $\ve{x}_{\text{a}} = \bPhi_{\text{a}} \balpha_{\text{a}}$ and $\ve{x}_{\text{b}} = \bPhi_{\text{b}} \balpha_{\text{b}}$ are referred to as {\em layers}. This setting covers a variety of situations. $\bPhi_{\aa}$ and $\bPhi_{\bb}$ may be equal with $\balpha_{\aa}$ and $\balpha_{\bb}$ having a different structure. For example, $\balpha_{\aa}$ may follow a GCM like before and $\balpha_{\bb}$ may be given a sparsity-inducing prior. In such a case, multi-layer LRTFS offers a synthesis perspective to sparse + low-rank spectrogram decompositions, such as those presented in \cite{Huang2012,Chen2013,Sun2014} which propose variants of robust principal component analysis (RPCA) \cite{Candes2009} for spectral unmixing. Even more interestingly, the time-frequency dictionaries $\bPhi_{\text{a}}$ and $\bPhi_{\text{b}}$ may be chosen with different t-f resolutions. This yields so-called hybrid or morphological decompositions \cite{Daudet2002,starck2005morphological}, in which each layer may capture specific resolution-dependent structures. A typical audio example is transient + tonal decomposition: transient components are by nature adequately represented by a t-f dictionary with short time resolution while tonal components (such as the sustained parts of musical notes) are better represented by a t-f dictionary with larger time resolution (and as a consequence, finer frequency resolution). A variety of priors can be considered for $\balpha_{\aa}$ and $\balpha_{\bb}$, such as frequency grouping for the transient synthesis coefficients and temporal grouping for the tonal synthesis coefficients \cite{Kowalski2009a}.

	\subsection{Sparse and low-rank time-frequency synthesis}

	We consider a special case of multi-layer LRTFS that illustrates the potential of the synthesis approach. We present the methodology in the complex case for simplicity, but the results can readily be adapted to the real case following the procedure described in Section~\ref{sec:lrtfs}.

	\subsubsection{Model}
	Let $\bPhi_{\aa}$ and $\bPhi_{\bb}$ be time-frequency dictionaries consisting of atoms $\bphi_{fn}^{\aa}(t)$ and $\bphi_{fn}^{\bb}(t)$ with common dimension $T$ and t-f pavings of size $F_{\aa} \times N_{\aa}$ and $F_{\bb} \times N_{\bb}$, respectively. We consider the following model, for $t=1,\ldots,T$:
	\begin{align}
		x(t) &= \sum_{f=1}^{F_{\aa}} \sum_{n=1}^{N_{\aa}} \alpha_{fn}^{\aa} \phi_{fn}^{\aa}(t) + \sum_{f=1}^{F_{\bb}} \sum_{n=1}^{N_{\bb}} \alpha^{\bb}_{fn} \phi^{\bb}_{fn}(t) + e(t) \label{eqn:sigh} \\
		\alpha^{\aa}_{fn} &\sim N_{c}(0,[\W \H]_{fn}), f=1,\ldots,F_{\aa}, n=1,\ldots,N_{\aa} \label{eqn:alphaa} \\
		\alpha^{\bb}_{fn} &\sim N_{c}(0,v^{\bb}_{fn}), f=1,\ldots,F_{\bb}, n=1,\ldots,N_{\bb} \label{eqn:alphab} \\
		e(t) &\sim N_{c}(0,\lambda) \label{eqn:resh}
	\end{align}
	where $\{ \alpha^{\aa}_{fn} \}$ and $\{ \alpha^{\bb}_{fn} \}$ are the complex-valued synthesis coefficients, $\W$ and $\H$ are nonnegative matrices of sizes $F_{\aa} \times K$ and $K \times N_{\aa}$, respectively, $\{ v^{\bb}_{fn} \}$ are nonnegative variance parameters and $e(t)$ is an additive complex-valued residual term.	Eq.~\eqref{eqn:sigh} is nothing but the scalar form of Eq.~\eqref{eqn:multil}. Eq.~\eqref{eqn:alphaa} defines a GCM, while Eq.~\eqref{eqn:alphab} defines the sparse-inducing prior that is used in SBL. Like before, we denote by $\ve{v}^{\aa}$ and $\ve{v}^{\bb}$ the column vectors with coefficients $[\W \H]_{fn}$ and $v^{\bb}_{fn}$, respectively. Both $\ve{v}^{\aa}$ and $\ve{v}^{\bb}$ are parameters of a hierarchical variance model. Notice however how $\ve{v}^{\bb}$ is a {\em free} parameter, while $\ve{v}^{\aa}$ is structured through $\W$ and $\H$. Overall, Eq.~\eqref{eqn:sigh}-\eqref{eqn:resh}, define a multi-layer LRTFS model with latent low-rank t-f structure for layer $\ve{x}_{\aa}$ and latent sparse t-f structure for layer $\ve{x}_{\bb}$.

\new{Note that a wavelet dictionary could alternatively be used to represent the sparse layer: the estimation procedure presented next would apply in the exact same way. LRTFS however has to be supported by a regular t-f lattice and is not compatible with wavelets (in which time resolution decreases with frequency). LRTFS can however accommodate constant-$Q$ t-f representations (constant time resolution, logarithmic frequency resolution) provided it can be inverted (near-accurate synthesis operator are proposed in  \cite{Brown1991,Fitzgerald2006}).}

	\subsubsection{Estimation}
	
	The negative log-likelihood of the data and parameters in model~\eqref{eqn:sigh}-\eqref{eqn:resh} is given by
	\bal{  %\label{eqn:cjl_hybrid0}
	& - \log p(\ve{x},\balpha_\aa,\balpha_\bb  |\W,\H,\ve{v}^{\bb})  = \frac{1}{\lambda} \| \ve{x} - \bPhi_\aa \balpha_\aa - \bPhi_\bb \balpha_\bb \|_2^2 \nonumber\\
	& \quad + D_{\text{IS}}(|\balpha_\aa|^2|\ve{v}_{\aa}) +  \log(|\balpha_\aa|^2) \nonumber\\
	& \quad + D_{\text{IS}}(|\balpha_\bb|^2|\ve{v}_{\bb}) +  \log(|\balpha_\bb|^2) + cst
	}
	where $cst =  T \log \lambda + (T+ F_{\aa}N_{\aa} + F_{\bb}N_{\bb}) \log \pi$. Unfortunately, and similarly to the difficulty of estimating $\lambda$ raised in Section~\ref{sec:mjle}, MJLE fails to evenly distribute the signal variance onto the two layers, and one of the two layers takes it all in practice. Such a problem can be mitigated using MMLE instead of MJLE, but again, MMLE is too costly in our setting. To solve this issue we introduce an extra hyper-parameter $\mu$ that balances the contributions of each layer and propose to optimize the following objective
	\bal{  \label{eqn:cjl_hybrid0}
	C_\text{SLR}(\btheta) \defequal& \frac{1}{\lambda} \|\ve{x} - \bPhi_\aa \balpha_\aa - \bPhi_\bb \balpha_\bb \|_2^2 \\
	& + \mu \left[ D_{\text{IS}}(|\balpha_\aa|^2|\ve{v}_{\aa}) +  \log(|\balpha_\aa|^2) \right] \nonumber\\
	& + (1-\mu) \left[ D_{\text{IS}}(|\balpha_\bb|^2|\ve{v}_{\bb}) +  \log(|\balpha_\bb|^2) \right] + cst,
	}
	where  $0\le\mu\le1$, $\btheta = \{\balpha_{\aa},\balpha_{\bb}, \W,\H, \ve{v}_{\bb} \} $ is the set of  latent variables and parameters and SLR stands for ``sparse + low-rank''.

	We may again find a stationary point of $C_\text{SLR}(\btheta)$ by alternate minimization. The update of $\ve{v}_{\bb}$ is trivially given by $\ve{v}_{\bb} = |\balpha_{\bb}|^{2}$. The update of $\W$ and $\H$ amounts to finding an IS-NMF of the synthesis spectrogram $|\alpha^{\aa}_{fn}|^{2}$ like in Algorithm~\ref{alg:LRTFS}. The synthesis coefficients $\balpha_{\aa}$ and $\balpha_{\bb}$ may be updated jointly via ridge regression over the joint dictionary $[ \bPhi_{\aa}, \bPhi_{\bb}]$. This leads to the following updates
	\bal{
	\hat{\ve{e}}^{(j)} & = \ve{x} - \bPhi_{\aa} \balpha_{\aa}^{(j)} - \bPhi_{\bb} \balpha_{\bb}^{(j)} \label{eqn:hres} \\
	\balpha_{\aa}^{(j+1/2)} &= \balpha_{\aa}^{(j)} + \frac{1}{L} \bPhi_{\aa}\hermconj \hat{\ve{e}}^{(j)} \label{eqn:hdesc1}  \\
	\balpha_{\bb}^{(j+1/2)} &= \balpha_{\bb}^{(j)} + \frac{1}{L} \bPhi_{\bb}\hermconj \hat{\ve{e}}^{(j)} \label{eqn:hdesc2}  \\
	\balpha_{\aa}^{(j+1)} &= \frac{\ve{v}_{\aa}}{\ve{v}_{\aa}+ {\lambda}/{L}} \circ \balpha_{\aa}^{(j+1/2)} \label{eqn:hshrink1}\\
	\balpha_{\bb}^{(j+1)} &= \frac{\ve{v}_{\bb}}{\ve{v}_{\bb}+ {\lambda}/{L}} \circ \balpha_{\bb}^{(j+1/2)} \label{eqn:hshrink2}
	}
	where the inverse-step size should satisfy $L \ge \| [ \bPhi_{\aa}, \bPhi_{\bb}] \|_{2}^{2}$. A convenient choice is $L = \|\bPhi_{\aa}\|^2_{2}+\|\bPhi_{\bb}\|^2_{2}$. Eq.~\eqref{eqn:hres} computes the current residual, Eqs.~\eqref{eqn:hdesc1} and~\eqref{eqn:hdesc2} produce a step in the descent direction and Eqs.~\eqref{eqn:hshrink1} and \eqref{eqn:hshrink2} shrink the resulting iterates. 
	
%	 $L\leq \|\bPhi_{aa}\|^2+\|\bPhi_{\bb}\|^2$ is the spectral norm of $[ \bPhi_{\aa}, \bPhi_{\bb}]\hermconj [ \bPhi_{\aa}, \bPhi_{\bb}] $. 

	\subsection{Example} \label{sec:exhlrtfs}

	We use exactly the same data and setting as in Section~\ref{sec:exlrtfs} but we now add a sparse layer $\bPhi_{\bb} \balpha_{\bb} $ to the LRTFS layer. $\bPhi_{\bb}$ is set to be a tight Gabor dictionary built on a Hann window of 128 samples ($6$ ms) with $50\%$ overlap. $\bPhi_{\aa}$ is set as in Section~\ref{sec:exlrtfs}. The parameter $\mu$ was experimentally fixed to $\mu=0.05$, and $\lambda$ was again chosen among logarithmically spaced vales. Fig.~\ref{fig:multilrtfs} displays the 10 latent components characterizing the tonal layer and the transient layer. The components of the tonal layer are similar to those obtained from the single-layer LRTFS decomposition of Fig.~\ref{fig:piano_1Layer}. The fourth component captures part of the hammer attacks (especially from the first, most energetic note) with the shortest resolution components relegated to the transient layer $x^{\bb}(t)$ as expected. Audio files are available online.

	\begin{figure}[!t]
		\centering
		{\small (a) Latent components of the tonal layer $x^{\aa}(t)$} \\		
		\includegraphics[width=\linewidth,trim={4cm 2cm 4cm 1cm},clip]{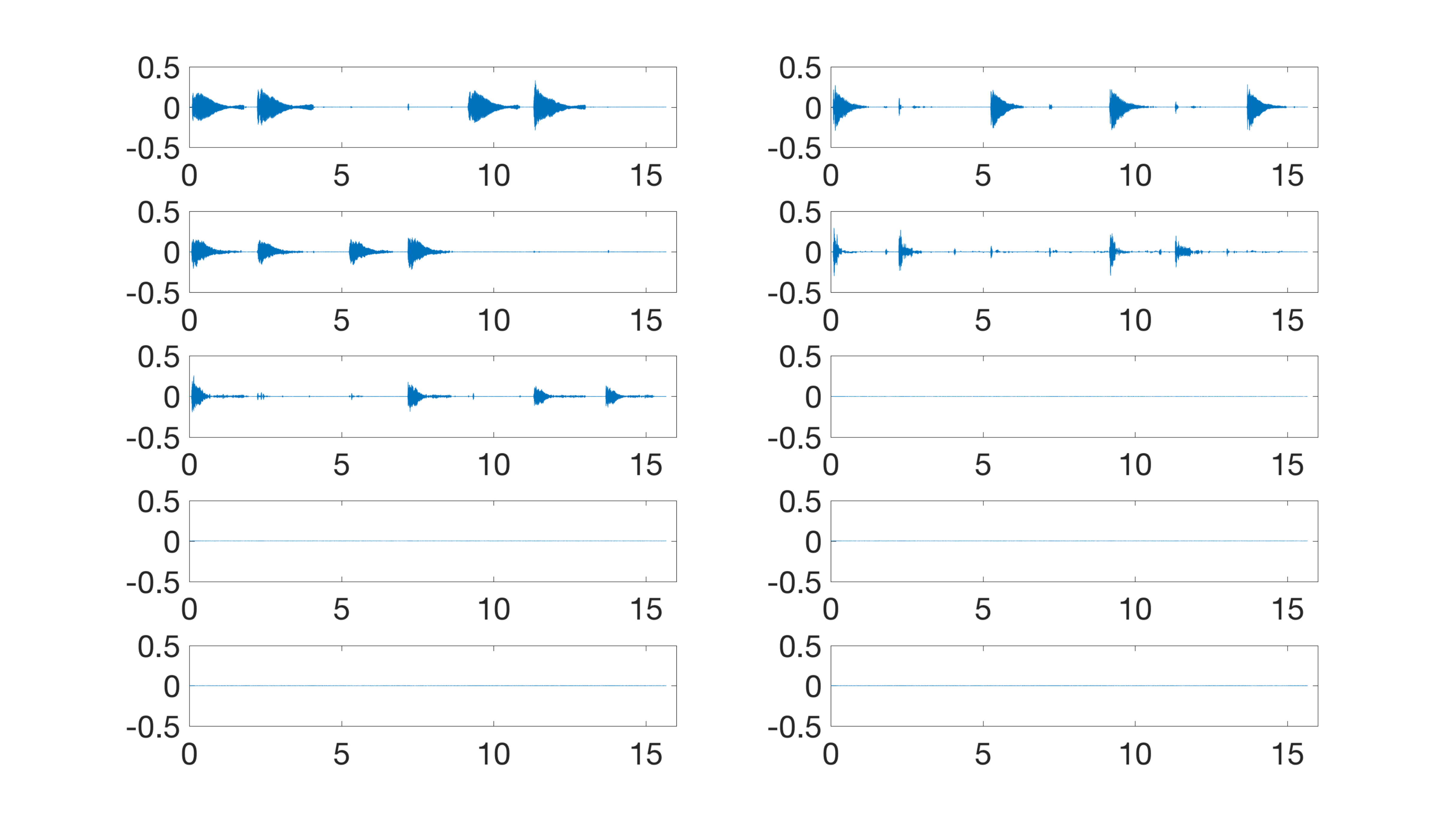} \\
		{\small (b) Transient layer $x^{\bb}(t)$} \\
		\includegraphics[width=\linewidth,trim={0cm 0.5cm 0 0},clip]{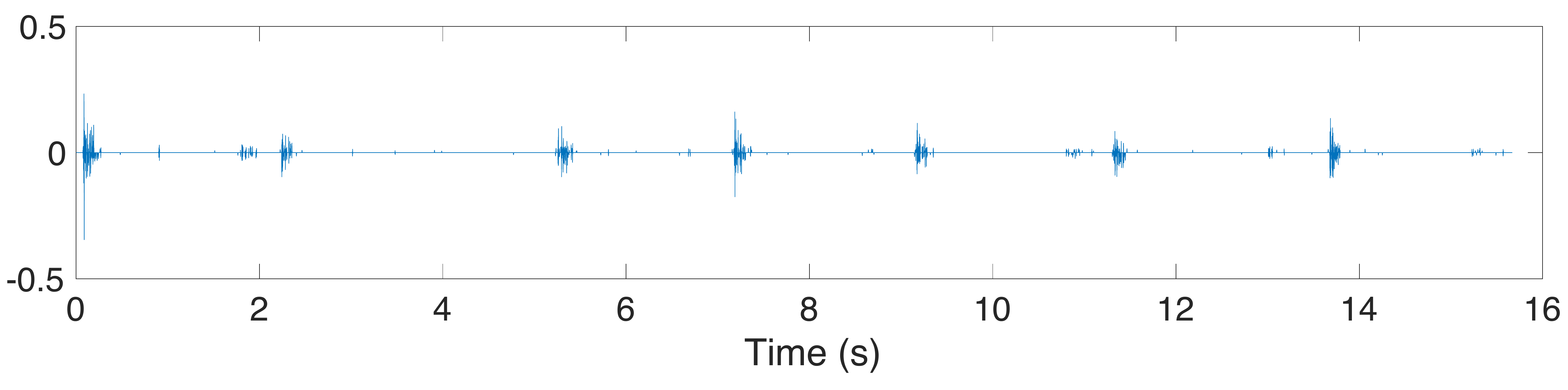}		
		\caption{Two-layer decomposition of the piano sequence displayed in Fig.~\ref{fig:piano_1Layer}.}
		\label{fig:multilrtfs}
	\end{figure}

	\section{Compressive LRTFS \label{sec:compress}}

	A striking advantage of LRTFS is that it may be used as a source model in inverse problems. For instance, LRTFS has been used in multichannel source separation in \cite{Leglaive2017}. We here consider compressive sensing (CS) in which a source signal $x(t)$ must be recovered from $S < <T$ random projections. Traditionally, CS exploits the sparsity of the synthesis coefficients of $x(t)$ onto a suitable dictionary. In this section we show that sparsity can be efficiently replaced with low-rankness, \new{for the class of signals considered}. 
	
	\subsection{Model}

	Let us denote by $\ve{x} \in \CC^T $ the vector source signal. The source is assumed to be sensed through the given linear operator $\ve{A} \in \CC^{S \times T}$ ($S < T$), with output $\ve{b} \in \CC^{S}$. We assume the following observation model:
	\bal{
	\ve{b} &  = \ve{A}\ve{x} + \ve{e} \\
	&   = \ve{A} \ve{\bPhi} \balpha+ \ve{e} 
	}
	where $\ve{\bPhi} \in \CC^{T \times M}$ is a given dictionary, $\balpha$ are the synthesis coefficients of $\ve{x}$, and $\ve{e}$ is a residual term that accounts for noise or model errors. Where traditional CS assumes some form of sparsity for $\balpha$, we assume the synthesis coefficients to have the LRTFS low-rank structure described by Eq.~\eqref{eqn:alpha}. Like in traditional CS settings, we assume $\ve{A}$ to be a random matrix. Finally, we assume $\ve{e}$ to follow a complex Gaussian distribution like in Eq.~\eqref{eqn:res}.

	\subsection{Estimation}

	MJLE amounts to minimizing the following objective function:
	\begin{align}
		& C_{\text{\text{CS}}}(\balpha,\W,\H) = - \log p(\ve{b}, \balpha | \W, \H, \lambda) \\
		&  = \frac{1}{\lambda} \|\ve{b} - \bA\bPhi\balpha\|_2^2 \nonumber + D_{\text{IS}}(|\balpha|^2|\W\H) + \log(|\balpha| ^2) + cst\label{eqn:cjl_inv}
	\end{align}
	where $cst = T \log \lambda + (T+ M) \log \pi$. The problem of optimizing $C_{\text{\text{CS}}}(\balpha,\W,\H)$ is equivalent to the one of optimizing $C_{\text{\text{JL}}}(\balpha,\W,\H)$ given by Eq.~\eqref{eqn:cjl}. In the complex case, the methodology developed in Section~\ref{sec:algo} can be readily applied by replacing $\bPhi$ with $\ve{M} = \ve{A} \bPhi$. The spectral norm of $\ve{M}$ may be difficult to derive or compute and we may set $L = \| \ve{A} \|_{2}^{2} \| \ve{\bPhi} \|_{2}^{2} $ thanks to the inequality
	\bal{
	\| \ve{A} \bPhi \|_{2}^{2} \le \| \ve{A} \|_{2}^{2} \| \ve{\bPhi} \|_{2}^{2}.
	}
	In the real case, i.e, when $\ve{x} \in \RR^T$, the methodology developed in Section~\ref{sec:lrtfs-real} may again be applied by assuming $\ve{A} \in \RR^{S \times T}$ and replacing $\buPhi$ with $\ve{M} = \ve{A} \buPhi$. Posterior to estimation, an estimate of the original source is given by $\hat{\ve{x}} = \bPhi \hat{\balpha}$.

Note that we have addressed compressive sampling of real or complex-valued signals by exploiting a latent NMF-type \mbox{t-f} structure, which is different from compressive sampling of non-negative signals, a topic addressed for example in \cite{OGrady2008}.

	\subsection{Example}
	
	We evaluate the recovery accuracy of the piano sequence used in Sections~\ref{sec:exlrtfs} and~\ref{sec:exhlrtfs} using a number of measurements $S$ varying increasingly from $T/100$ to $T/10$. For this experiment, the length of the sequence remains $15.6$~s but the sampling rate has been fixed at 11025 Hz because of the memory and computational complexities. The Gabor parameters have been adjusted accordingly with a Hann window of length $512$~samples ($46$~ms) with $50\%$ overlap.

	We compare CS recovery methods based on LRTFS, SBL and $\ell_{1}$ regularization, using a common alternating minimization setting (only the shrinkage or thresholding operators are changed). Note that we here consider type-I SBL (equivalent to MJLE) and not type-II (which again does not scale with the dimensions of our problem). The algorithms are initialized with $ \balpha =\mathbf{0}_{M \times 1}$. The first IS-NMF step of the LRTFS estimation was initialized with the absolute value of the complex-SVD as explained in Section~\ref{sec:algo}. LRTFS was applied with $K=10$ and the hyper-parameter $\lambda$ was incrementally decreased from $10^3$ to $10^{-2}$. \new{In addition, we provide the performance results of two oracles. In the first oracle, the vector of variances in Eq.~\eqref{eqn:prioralphR} is set to the power spectrogram of the ground-truth uncompressed signal $x(t)$ (updates of $\W$, $\H$ are thus removed). In the second oracle, we set the matrices $\W$, $\H$ to their estimates returned by IS-NMF applied to the power spectrogram of $x(t)$. These two oracles allow one to evaluate the remaining gap between adaptive and optimal CS recovery.}

%	As no denoising effect are expected, we choose the hyper-parameter $\lambda\rightarrow 0$. More precisely, the warm restart strategy is employed  as in Section~\ref{sec:algo} with $\lambda$ taking $30$ values on a logarithmic scale from $10^3$ to $10^{-2}$.

Estimation accuracy was measured by means of output SNR. The results are displayed in Fig.~\ref{fig:CS} and show that LRTFS-based recovery improves accuracy by several dBs as compared to sparsity-based methods. \new{This means that for this type of signals which are endowed with a strong low-rank t-f structure, there is a significant gain in exploiting low-rankness instead of mere unstructured sparsity for CS}. Such a recovery approach is made possible thanks to the generative design of LRTFS. Fig.~\ref{fig:piano_CS_comp} displays the estimated components $\hat{\ve{c}}_{k}$ returned by LRFTS. It is interesting to note that only $4$ components are meaningful. The first two notes are well recovered, like in the experiment of Section~\ref{sec:piano1}, see Fig.~\ref{fig:piano_1Layer}~(b), while the two other notes are mixed in the third component. The fourth component still captures some transient information. We also run experiments for various values of the rank $K$ in the case $S=5\,T/100$. The recovery results appeared very robust to this parameter. For $K \in \{ 5,8,10,15,20,30 \}$ the largest difference in the output SNRs was less than $0.5$~dB. \new{This robustness is partially explained by the fact that LRTFS tends to shrink irrelevant components, as explained by Eq.~\eqref{eqn:shrink} and illustrated by Fig.~\ref{fig:piano_1Layer}. We believe that the deterministic initialization provided by SVD is another explanation.}

%\new{A possible explanation for this robustness with respect to the rank $K$ is the deterministic initialization by complex SVD (see end of Section~\ref{sec:algo}) {\bf sparsification ?}.}

	Finally, we run the same CS experiment using the first $12$~s of the song {\em Mamavatu} from Susheela Raman. The excerpt contains acoustic guitar and drums. Output SNRs are displayed on Fig.~\ref{fig:mamavatu_CS}. Again, LRTFS recovery outperforms $\ell_{1}$ regularization and SBL by several dBs which confirms the potential of the proposed model for audio inverse problems, \new{where t-f low-rankness is a valid assumption.}

	%This robustness can be explained by the deterministic initialization by complex-SVD. 

	\begin{figure}[t]
		\centering
		\includegraphics[width=\linewidth,trim={3cm 0cm 4cm 1cm},clip]{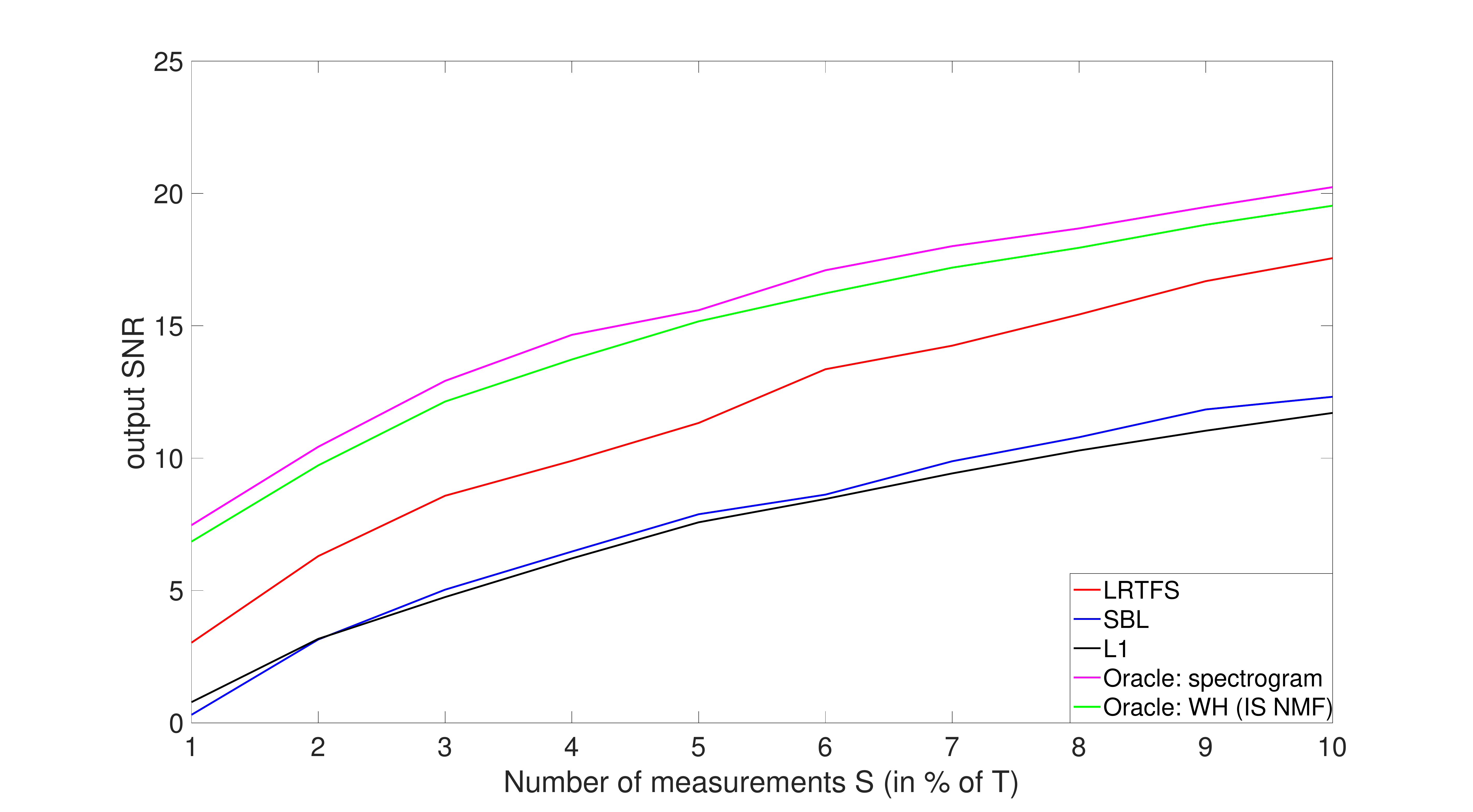}
		\caption{Recovery of a compressively sensed piano sequence using LRTFS, SBL and $\ell_{1}$ regularization, compared with two oracles.}
		\label{fig:CS}
	\end{figure}

	\begin{figure}[t]
		\centering
		\includegraphics[width=\linewidth,trim={4cm 2cm 4cm 1cm},clip]{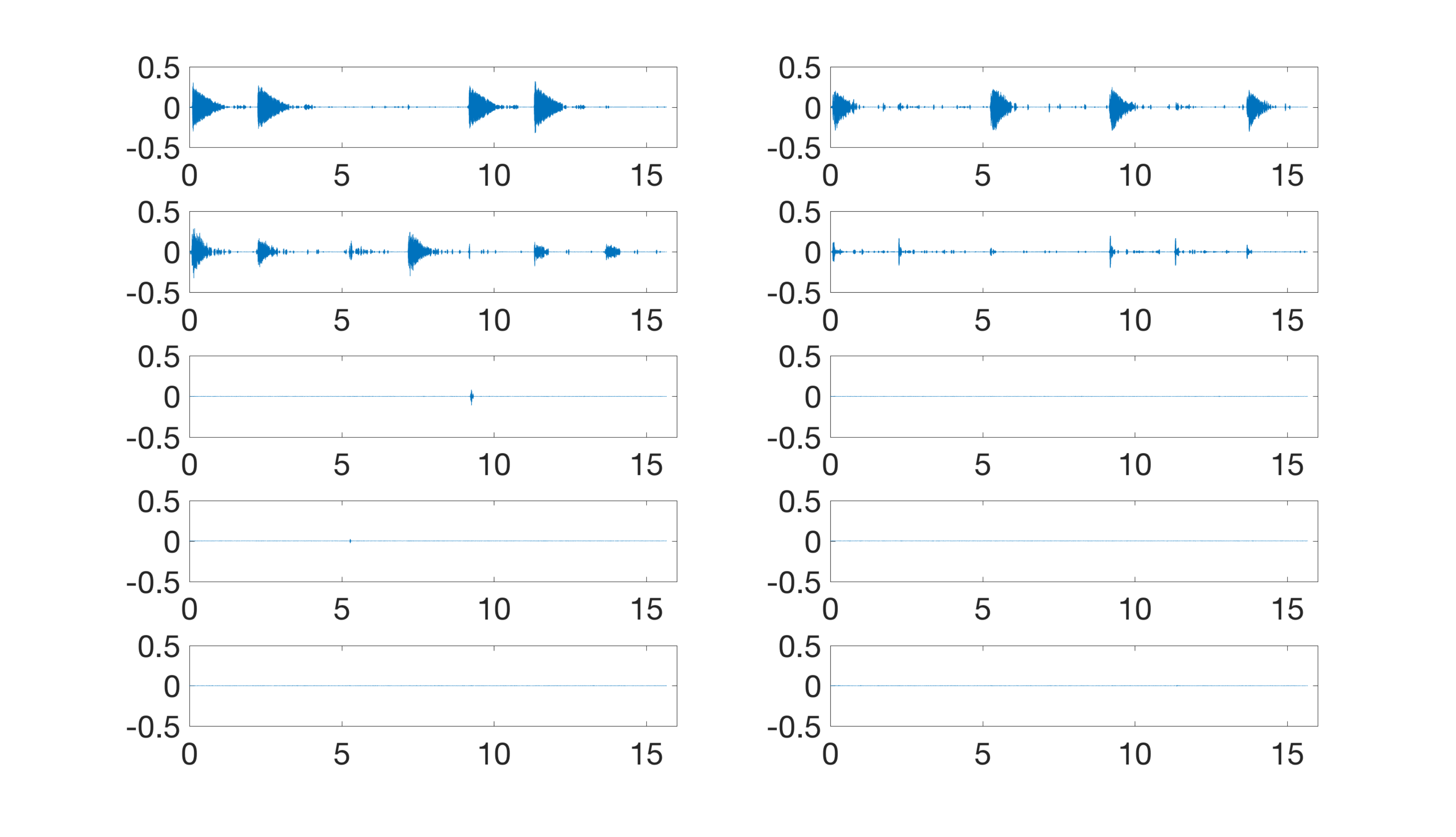} \\
		\caption{Latent components of the compressively sensed piano sequence recovered by LRTFS. The temporal components are displayed by decreasing energy (from left to right and top to bottom).}
		\label{fig:piano_CS_comp}
	\end{figure}

	\begin{figure}[t]
		\centering
		\includegraphics[width=\linewidth,trim={3cm 0cm 4cm 1cm},clip]{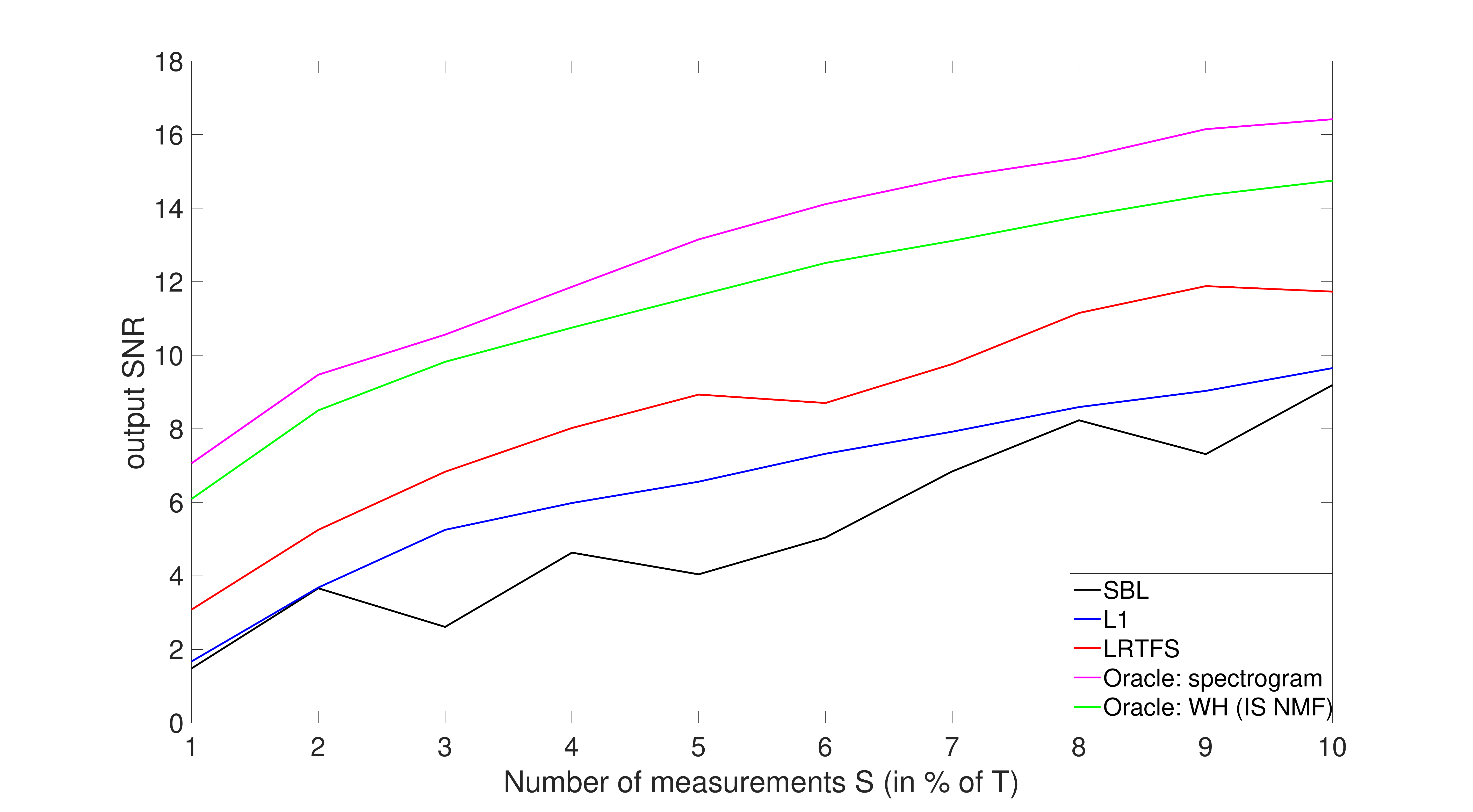}
		\caption{Recovery of the compressively sensed Mamavatu sequence using LRTFS, SBL and $\ell_{1}$ regularization, compared with two oracles.}
		\label{fig:mamavatu_CS}
	\end{figure}

	\section{Conclusion}	\label{sec:conclusion}

	We have presented a new modeling paradigm that bridges t-f synthesis modeling and traditional analysis-based approaches. The proposed generative model allows in turn to design more sophisticated multi-layer representations that can efficiently capture diverse forms of structure. Additionally, the generative modeling allows to exploit NMF-like structure for compressive sensing which, to the best of our knowledge, is entirely new. Maximum joint likelihood estimation in the proposed models can be efficiently addressed using state-of-the-art iterative shrinkage and NMF algorithms. They can be efficiently implemented thanks to dedicated time-frequency analysis/synthesis packages. In this paper, we also addressed the modeling and decomposition of real signals in a rigorous way, which was missing from our preliminary contributions and appeared more tricky than initially expected.

	The MLJE objective function~\eqref{eqn:cjl} induced by the proposed generative modeling suggests more general problems of the form
	\bal{
	C(\balpha,\W,\H,\lambda) = \frac{1}{\lambda} \|\ve{x} - \bPhi\balpha\|_2^2 + D(|\balpha|^p|\ve{v}) 
	}
	where $\ve{v}=\text{vec}[\W \H]$, $D(\cdot|\cdot)$ is an arbitrary divergence between nonnegative numbers and $p$ is an arbitrary exponent. $D=D_{\text{IS}}$ and $p=2$ follow naturally from the GCM assumptions but other choices could be more suitable for other families of signals or images. Such problems do not seem to have been addressed yet in the literature and offer stimulating optimization problems. The exact reconstruction case $\lambda = 0$ is also very interesting in itself. Another challenging line of research is the design of workable large-scale optimization algorithms for type-II maximum marginal likelihood estimation. As known from \cite{Wipf2004}, such an estimator would be robust to the joint estimation of $\lambda$ and $\ve{v}$, something in which MJLE fails in practice. \new{The low-rank structure used in Eq.~\eqref{eqn:alpha} to model the variance of the synthesis coefficients could be changed for more complex structures, such as neural architectures. This has been considered to model STFT synthesis coefficients in audio applications \cite{Bando2018}, using for example variational auto-encoders for training \cite{Kingma2004}. The framework presented in this paper can readily accommodate such variants, in particular in the multi-layer setting in which a layer can be assigned a pre-trained variance (for a specific class of signals such as speech) and another layer can be endowed with a free adaptive low-rank variance.}

\section*{Acknowledgment}
We are grateful to the reviewers for their valuable comments. C\'edric F\'evotte acknowledges support from the European Research Council (ERC) under the European Union's Horizon 2020 research and innovation programme under grant agreement No 681839 (project FACTORY).
	
%	\bibliographystyle{IEEEtran}
%	\bibliography{mk_biblio.bib,biblio_perso.bib,biblio_general.bib}

\begin{thebibliography}{10}
\providecommand{\url}[1]{#1}
\csname url@samestyle\endcsname
\providecommand{\newblock}{\relax}
\providecommand{\bibinfo}[2]{#2}
\providecommand{\BIBentrySTDinterwordspacing}{\spaceskip=0pt\relax}
\providecommand{\BIBentryALTinterwordstretchfactor}{4}
\providecommand{\BIBentryALTinterwordspacing}{\spaceskip=\fontdimen2\font plus
\BIBentryALTinterwordstretchfactor\fontdimen3\font minus
  \fontdimen4\font\relax}
\providecommand{\BIBforeignlanguage}[2]{{%
\expandafter\ifx\csname l@#1\endcsname\relax
\typeout{** WARNING: IEEEtran.bst: No hyphenation pattern has been}%
\typeout{** loaded for the language `#1'. Using the pattern for}%
\typeout{** the default language instead.}%
\else
\language=\csname l@#1\endcsname
\fi
#2}}
\providecommand{\BIBdecl}{\relax}
\BIBdecl

\bibitem{Smaragdis2014}
\BIBentryALTinterwordspacing
P.~Smaragdis, C.~F{\'e}votte, G.~Mysore, N.~Mohammadiha, and M.~Hoffman,
  ``Static and dynamic source separation using nonnegative factorizations: {A}
  unified view,'' \emph{IEEE Signal Processing Magazine}, vol.~31, no.~3, pp.
  66--75, May 2014.
\BIBentrySTDinterwordspacing

\bibitem{vir07}
T.~Virtanen, ``Monaural sound source separation by non-negative matrix
  factorization with temporal continuity and sparseness criteria,'' \emph{IEEE
  Transactions on Audio, Speech and Language Processing}, vol.~15, no.~3, pp.
  1066--1074, Mar. 2007.

\bibitem{ieee_asl10}
\BIBentryALTinterwordspacing
A.~Ozerov and C.~F\'evotte, ``Multichannel nonnegative matrix factorization in
  convolutive mixtures for audio source separation,'' \emph{IEEE Transactions
  on Audio, Speech and Language Processing}, vol.~18, no.~3, pp. 550--563, Mar.
  2010. 
\BIBentrySTDinterwordspacing

\bibitem{elad2007analysis}
M.~Elad, P.~Milanfar, and R.~Rubinstein, ``Analysis versus synthesis in signal
  priors,'' \emph{Inverse problems}, vol.~23, no.~3, p. 947, 2007.

\bibitem{balazs2013adapted}
P.~Balazs, M.~Doerfler, M.~Kowalski, and B.~Torr{\'e}sani, ``Adapted and
  adaptive linear time-frequency representations: a synthesis point of view,''
  \emph{IEEE Signal Processing Magazine}, vol.~30, no.~6, pp. 20--31, 2013.

\bibitem{sprechmann2013supervised}
P.~Sprechmann, R.~Litman, T.~B. Yakar, A.~M. Bronstein, and G.~Sapiro,
  ``Supervised sparse analysis and synthesis operators,'' in \emph{Advances in
  Neural Information Processing Systems}, 2013, pp. 908--916.

\bibitem{nam2013cosparse}
S.~Nam, M.~E. Davies, M.~Elad, and R.~Gribonval, ``The cosparse analysis model
  and algorithms,'' \emph{Applied and Computational Harmonic Analysis},
  vol.~34, no.~1, pp. 30--56, 2013.

\bibitem{neco09}
\BIBentryALTinterwordspacing
C.~F\'evotte, N.~Bertin, and J.-L. Durrieu, ``Nonnegative matrix factorization
  with the {I}takura-{S}aito divergence. {W}ith application to music
  analysis,'' \emph{Neural Computation}, vol.~21, no.~3, pp. 793--830, Mar.
  2009. 
\BIBentrySTDinterwordspacing

\bibitem{Kameoka2015}
H.~Kameoka, ``Multi-resolution signal decomposition with time-domain
  spectrogram factorization,'' in \emph{Proc.~IEEE International Conference on
  Acoustics, Speech and Signal Processing (ICASSP)}, 2015.

\bibitem{Kameoka2017}
------, ``Complex {NMF} with the generalized {K}ullback-{L}eibler divergence,''
  in \emph{Proc.~IEEE International Conference on Acoustics, Speech and Signal
  Processing (ICASSP)}, 2017.

\bibitem{Liutkus2011}
A.~Liutkus, R.~Badeau, and G.~Richard, ``Gaussian processes for underdetermined
  source separation,'' \emph{IEEE Transactions on Signal Processing}, vol.~59,
  no.~7, pp. 3155--3167, July 2011.

\bibitem{Yoshii2013}
K.~Yoshii, R.~Tomioka, D.~Mochihashi, and M.~Goto, ``Infinite positive
  semidefinite tensor factorization for source separation of mixture signals,''
  in \emph{Proc.~International Conference on Machine Learning (ICML)}, 2013.

\bibitem{Turner2014}
R.~E. Turner and M.~Sahani, ``Time-frequency analysis as probabilistic
  inference,'' \emph{IEEE Transactions on Signal Processing}, vol.~62, no.~23,
  pp. 6171--6183, Dec 2014.

\bibitem{Liutkus2017}
A.~Liutkus and K.~Yoshii, ``A diagonal plus low-rank covariance model for
  computationally efficient source separation,'' in \emph{Proc.~IEEE
  International Workshop on Machine Learning for Signal Processing (MLSP)},
  2017.

\bibitem{Yoshii2018}
K.~Yoshii, ``Correlated tensor factorization for audio source separation,'' in
  \emph{Proc.~IEEE International Conference on Acoustics, Speech, and Signal
  Processing (ICASSP)}, 2018.

\bibitem{nips14}
\BIBentryALTinterwordspacing
C.~F\'evotte and M.~Kowalski, ``Low-rank time-frequency synthesis,'' in
  \emph{Advances in Neural Information Processing Systems (NIPS)}, Dec. 2014.
\BIBentrySTDinterwordspacing

\bibitem{Fevotte2015a}
\BIBentryALTinterwordspacing
------, ``Hybrid sparse and low-rank time-frequency signal decomposition,'' in
  \emph{Proc.~European Signal Processing Conference (EUSIPCO)}, Nice, France,
  Sep. 2015.
\BIBentrySTDinterwordspacing

\bibitem{lee99}
D.~D. Lee and H.~S. Seung, ``Learning the parts of objects with nonnegative
  matrix factorization,'' \emph{Nature}, vol. 401, pp. 788--791, 1999.

\bibitem{tip01}
M.~E. Tipping, ``Sparse {B}ayesian learning and the relevance vector machine,''
  \emph{Journal of Machine Learning Research}, vol.~1, pp. 211--244, 2001.

\bibitem{Wipf2004}
D.~P. Wipf and B.~D. Rao, ``Sparse bayesian learning for basis selection,''
  \emph{IEEE Transactions on Signal Processing}, vol.~52, no.~8, pp.
  2153--2164, Aug. 2004.

\bibitem{Figueiredo2003}
M.~Figueiredo and R.~Nowak, ``An {EM} algorithm for wavelet-based image
  restoration,'' \emph{IEEE Transactions on Image Processing}, vol.~12, no.~8,
  pp. 906--916, 2003.

\bibitem{Beck2009}
A.~Beck and M.~Teboulle, ``A fast iterative shrinkage-thresholding algorithm
  for linear inverse problems,'' \emph{SIAM Journal on Imaging Sciences},
  vol.~2, no.~1, pp. 183--202, 2009.

\bibitem{Chaari2011}
L.~Cha{\^a}ri, J.-C. Pesquet, A.~Benazza-Benyahia, and P.~Ciuciu, ``A
  wavelet-based regularized reconstruction algorithm for {SENSE} parallel {MRI}
  with applications to neuroimaging,'' \emph{Medical Image Analysis}, vol.~15,
  no.~2, pp. 185--201, 2011.

\bibitem{Florescu2014}
A.~Florescu, E.~Chouzenoux, J.-C. Pesquet, P.~Ciuciu, and S.~Ciochina, ``A
  majorize-minimize memory gradient method for complex-valued inverse
  problems,'' \emph{Signal Processing}, vol. 103, pp. 285--295, 2014.

\bibitem{chambolle2015convergence}
A.~Chambolle and C.~Dossal, ``On the convergence of the iterates of the “fast
  iterative shrinkage/thresholding algorithm”,'' \emph{Journal of
  Optimization Theory and Applications}, vol. 166, no.~3, pp. 968--982, 2015.

\bibitem{Ben-Tal2004}
A.~N.~A. Ben-Tal and A.~Nemirovski, ``Optimization {III}: {C}onvex analysis,
  {N}onlinear programming theory, {S}tandard nonlinear programming
  algorithms,'' Lecture Notes, 2004.

\bibitem{betanmf}
\BIBentryALTinterwordspacing
C.~F\'{e}votte and J.~Idier, ``Algorithms for nonnegative matrix factorization
  with the beta-divergence,'' \emph{Neural Computation}, vol.~23, no.~9, pp.
  2421--2456, Sep. 2011.
\BIBentrySTDinterwordspacing

\bibitem{becker2015complex}
J.~Becker, M.~Menzel, and C.~Rohlfing, ``Complex {SVD} initialization for {NMF}
  source separation on audio spectrograms,'' \emph{Proc.~Deutsche Jahrestagung
  f{\" u}r Akustik (DAGA)}, 2015.

\bibitem{hale2008fixed}
E.~T. Hale, W.~Yin, and Y.~Zhang, ``Fixed-point continuation for
  $\ell_1$-minimization: Methodology and convergence,'' \emph{SIAM Journal on
  Optimization}, vol.~19, no.~3, pp. 1107--1130, 2008.

\bibitem{Pruuvsa2014}
Z.~Pru{\v{s}}a, P.~L. Sondergaard, N.~Holighaus, C.~Wiesmeyr, and P.~Balazs,
  ``The large time-frequency analysis toolbox 2.0,'' in \emph{Sound, Music, and
  Motion, Lecture Notes in Computer Science}.\hskip 1em plus 0.5em minus
  0.4em\relax Springer, 2014, pp. 419--442.

\bibitem{Huang2012}
P.-S. Huang, S.~D. Chen, P.~Smaragdis, and M.~Hasegawa-Johnson, ``Singing-voice
  separation from monaural recordings using robust principal component
  analysis,'' in \emph{Proc.~IEEE International Conference on Acoustics, Speech
  and Signal Processing (ICASSP)}, 2012.

\bibitem{Chen2013}
Z.~Chen and D.~P.~W. Ellis, ``Speech enhancement by sparse, low-rank, and
  dictionary spectrogram decomposition,'' in \emph{Proc.~IEEE Workshop on
  Applications of Signal Processing to Audio and Acoustics (WASPAA)}, 2013.

\bibitem{Sun2014}
\BIBentryALTinterwordspacing
C.~Sun, Q.~Zhu, and M.~Wan, ``A novel speech enhancement method based on
  constrained low-rank and sparse matrix decomposition,'' \emph{Speech
  Communication}, vol.~60, pp. 44--55, 2014. 
\BIBentrySTDinterwordspacing

\bibitem{Candes2009}
E.~J. Cand\`es, X.~Li, Y.~Ma, and J.~Wright, ``Robust principal component
  analysis?'' \emph{Journal of ACM}, vol.~58, no.~1, pp. 1--37, 2009.

\bibitem{Daudet2002}
L.~Daudet and B.~Torr\'esani, ``Hybrid representations for audiophonic signal
  encoding,'' \emph{Signal Processing}, vol.~82, no.~11, pp. 1595 -- 1617,
  2002.

\bibitem{starck2005morphological}
J.-L. Starck, Y.~Moudden, J.~Bobin, M.~Elad, and D.~Donoho, ``Morphological
  component analysis,'' in \emph{Optics \& Photonics}, 2005.

\bibitem{Kowalski2009a}
M.~Kowalski, ``Sparse regression using mixed norms,'' \emph{Applied and
  Computational Harmonic Analysis}, vol.~27, no.~3, pp. 303--324, 2009.

\bibitem{Brown1991}
J.~C. Brown, ``Calculation of a constant {Q} spectral transform,'' \emph{The
  Journal of the Acoustical Society of America}, vol.~89, no.~1, pp. 425--434,
  1991.

\bibitem{Fitzgerald2006}
D.~Fitzgerald, M.~Cranitch, and M.~T. Cychowski, ``Towards an inverse constant
  {Q} transform,'' in \emph{Proc. Audio Engineering Society Convention}, 2006.

\bibitem{Leglaive2017}
S.~Leglaive, R.~Badeau, and G.~Richard, ``Separating time-frequency sources
  from time-domain convolutive mixtures using non-negative matrix
  factorization,'' in \emph{Proc.~IEEE Workshop on Applications of Signal
  Processing to Audio and Acoustics (WASPAA)}, 2017.

\bibitem{OGrady2008}
P.~D. O'Grady and S.~T. Rickard, ``Compressive sampling of non-negative
  signals,'' in \emph{Proc.~ IEEE Workshop on Machine Learning for Signal
  Processing (MLSP)}, 2008.

\bibitem{Bando2018}
Y.~Bando, M.~Mimura, K.~Itoyama, K.~Yoshii, and T.~Kawahara, ``Statistical
  speech enhancement based on probabilistic integration of variational
  autoencoder and non-negative matrix factorization,'' in \emph{Proc.~IEEE
  International Conference on Acoustics, Speech, and Signal Processing
  (ICASSP)}, 2018.

\bibitem{Kingma2004}
D.~Kingma and M.~Welling, ``Auto-encoding variational {B}ayes,'' in
  \emph{Proc.~International Conference on Learning Representations}, 2004.

\end{thebibliography}
\newpage

%	[{\includegraphics[width=1in,height=1.25in,clip,keepaspectratio]{fig/portrait2016}}]
	\begin{IEEEbiography}{C\'edric F\'evotte} is a CNRS senior researcher at Institut de Recherche en Informatique de Toulouse (IRIT). Previously, he has been a CNRS researcher at Laboratoire Lagrange (Nice, 2013-2016) \& T\'el\'ecom ParisTech (2007-2013), a research engineer at Mist-Technologies (the startup that became Audionamix, 2006-2007) and a postdoc at University of Cambridge (2003-2006). He holds MEng and PhD degrees in EECS from \'Ecole Centrale de Nantes. His research interests concern statistical signal processing and machine learning, for inverse problems and source separation. He was a member of the IEEE Machine~Learning for Signal Processing technical committee (2012-2018) and is a member of SPARS steering committee since 2018. He has been a member of the editorial board of the IEEE Transactions on~Signal Processing since 2014, first as an associate editor and then as a senior area editor (from 2018). In 2014, he was the co-recipient of an IEEE Signal Processing Society Best Paper Award for his work on audio source separation using multichannel nonnegative matrix factorization. He is the principal investigator of the European Research Council project FACTORY (New paradigms for latent factor estimation, 2016-2021).
	\end{IEEEbiography}

	\begin{IEEEbiography}{Matthieu Kowalski}
		received the engineering degree in computer science from the Universit{\'e} de Technologie de Compi{\`e}gne in 2005, and the master degree in Mathematics Vision and Learning from the Ecole Normale Sup\'erieur, Cachan, the same year. He received the PhD degree in applied mathematics from the University of Provence in 2008. His thesis was axed on sparse time-frequency decompositions. He is now an associate professor at the University of Paris-Sud, in the L2S Lab, and his research focuses on Inverse Problems and structured sparse approximations. He is an elected member of the SPARS Steering committee since 2013. 
	\end{IEEEbiography}

\end{document}